%% file: iclr2024_conference.tex
\documentclass{article} %
\usepackage{iclr2024_conference,times}
\usepackage{graphicx}

\usepackage{fontawesome}
\usepackage{stackengine}
\usepackage{multirow}
\usepackage{dsfont}

\input{math_commands.tex}

\usepackage{xcolor}
\definecolor{cyan}{rgb}{0.0, 1.0, 1.0}
\definecolor{darkred}{rgb}{0.55, 0.0, 0.0}

\usepackage{hyperref}
\usepackage{url}
\usepackage{subcaption}
\usepackage{ifthen}

\usepackage{twoxtwogame}

\usepackage{tikz}
\usetikzlibrary{calc,shapes}
\newcommand{\tikzmark}[1]{\tikz[overlay,remember picture] \node (#1) {};}
\newcommand{\drawpermutationtikz}[4]{%
    \ifthenelse{\equal{#1}{north}}{%
        \edef\Aangle{90}%
        \edef\Ashortstart{8pt}%
        \edef\Ashortend{5pt}%
    }{
        \edef\Aangle{270}%
        \edef\Ashortstart{4pt}%
        \edef\Ashortend{1pt}%
    }
    \ifthenelse{\equal{#2}{north}}{%
        \edef\Bangle{90}%
        \edef\Bshortstart{8pt}%
        \edef\Bshortend{5pt}%
    }{
        \edef\Bangle{270}%
        \edef\Bshortstart{4pt}%
        \edef\Bshortend{1pt}%
    }
    \ifthenelse{\equal{#3}{north}}{%
        \edef\Cangle{90}%
        \edef\Cshortstart{8pt}%
        \edef\Cshortend{5pt}%
    }{
        \edef\Cangle{270}%
        \edef\Cshortstart{4pt}%
        \edef\Cshortend{1pt}%
    }
    \ifthenelse{\equal{#4}{north}}{%
        \edef\Dangle{90}%
        \edef\Dshortstart{8pt}%
        \edef\Dshortend{5pt}%
    }{
        \edef\Dangle{270}%
        \edef\Dshortstart{4pt}%
        \edef\Dshortend{1pt}%
    }
    \begin{tikzpicture}[overlay,remember picture]
        \draw[->,shorten >=\Ashortend,shorten <=\Ashortstart,out=\Aangle,in=\Aangle,distance=0.5cm] (A.#1) to (sA.#1);
        \draw[->,shorten >=\Bshortend,shorten <=\Bshortstart,out=\Bangle,in=\Bangle,distance=0.5cm] (B.#2) to (sB.#2);
        \draw[->,shorten >=\Cshortend,shorten <=\Cshortstart,out=\Cangle,in=\Cangle,distance=0.5cm] (C.#3) to (sC.#3);
        \draw[->,shorten >=\Dshortend,shorten <=\Dshortstart,out=\Dangle,in=\Dangle,distance=0.5cm] (D.#4) to (sD.#4);
    \end{tikzpicture}
}

\newcommand{\drawpermutation}[8]{%
    {a\tikzmark{sA}\tikzmark{#1}_1^1} ~~~~ {a\tikzmark{sB}\tikzmark{#2}_1^2} ~~~~ {a\tikzmark{sC}\tikzmark{#3}_2^1} ~~~~ {a\tikzmark{sD}\tikzmark{#4}_2^2}
    \drawpermutationtikz{#5}{#6}{#7}{#8}
}

\title{NfgTransformer: Equivariant Representation Learning for Normal-form Games}

\author{Siqi Liu\textsuperscript{$\dagger$,$\diamondsuit$}, Luke Marris\textsuperscript{$\dagger$,$\diamondsuit$}, Georgios Piliouras\textsuperscript{$\dagger$}, Ian Gemp\textsuperscript{$\dagger$}, Nicolas Heess\textsuperscript{$\dagger$} \\
\textsuperscript{$\dagger$}Google DeepMind, \textsuperscript{$\diamondsuit$}University College London \\
\texttt{\{liusiqi,marris,gpil,imgemp,heess\}@google.com}
}

\iclrfinalcopy %
\begin{document}

\maketitle

\begin{abstract}
Normal-form games (NFGs) are the fundamental model of {\em strategic interaction}. We study their representation using neural networks. We describe the inherent equivariance of NFGs --- any permutation of strategies describes an equivalent game --- as well as the challenges this poses for representation learning. We then propose the NfgTransformer\footnote{The model is open-sourced at \url{https://github.com/google-deepmind/nfg_transformer}.} architecture that leverages this equivariance, leading to state-of-the-art performance in a range of game-theoretic tasks including equilibrium-solving, deviation gain estimation and ranking, with a common approach to NFG representation. We show that the resulting model is interpretable and versatile, paving the way towards deep learning systems capable of game-theoretic reasoning when interacting with humans and with each other.
\end{abstract}

\section{Introduction}

Representing data in a learned embedding space, or representation learning, is one of the timeless ideas of deep learning. Foundational representation learning architectures \citep{lecun1995convolutional, hochreiter1997long, vaswani2017attention} have provided performance and generality, bringing together researchers who used to work on expert systems targeting narrow domains. Consider Convolutional Neural Networks (CNNs): by exploiting the translation invariance inherent to images, CNNs replaced feature descriptors underpinning tasks as diverse as image classification, segmentation and in-painting, leading to paradigm shifts across the field.

One area of research that has resisted this trend is game theory. Here, we see two striking similarities to classical computer vision research. First, active research topics such as equilibrium solving \citep{vlatakis2020no}, ranking \citep{balduzzi2018re}, social choice theory \citep{Anshelevich2021DistortionIS} and population learning \citep{lanctot2017unified} all focus on specialised solutions, despite sharing the common language of normal-form games (NFGs). Second, task-specific solutions suffer from fundamental limitations. The popular Elo ranking algorithm \citep{elo2008rating}, for instance, assigns a scalar rating to each player derived from an NFG between players. Although Elo ratings are designed to be predictive of match outcomes, they are poor predictors beyond transitive games --- the Elo score is simply too restrictive a representation to reflect cyclic game dynamics. Improvements to Elo followed \citep{bertrand2023limitations}, but all relied on engineered feature descriptors, instead of learning. In equilibrium-solving, computing exact Nash equilibria is intractable beyond two-player zero-sum games \citep{daskalakis2009complexity} yet approximate solvers are non-differentiable, take non-deterministic amount of time to converge, struggle to parallelise, and can fail. These fundamental limitations have indirect consequences too. An entire line of works in tabular multiagent reinforcement learning (RL) \citep{littman2001friend, hu2003nash, greenwald2003correlated} relied on equilibrium solving as part of their learning rules --- an NFG is constructed from agents' Q-tables in each state, whose equilibrium informs subsequent policy updates. Unfortunately, reviving these ideas in the context of {\em deep} RL has been challenging, if not impossible, as it requires equilibrium solving as a subroutine in between every gradient update. 

Indeed, we are not the first to recognise these limitations. Several recent works incorporated representation learning {\em implicitly}, in {\em narrow domains} \citep{marris2022turbocharging, duan2023nash, vadori2023ordinal}. 
We address these limitations {\em explicitly} and {\em in generality}. Our goal is to develop principled, general-purpose representation of NFGs that can be used in a wide range of game-theoretic applications. We ask 1) which property, if any, could we leverage for efficient representation learning of NFGs without loss of generality and 2) can we expect performance in a range of game-theoretic tasks using a common approach to representing NFGs. We answer both questions affirmatively and propose NfgTransformer, a general-purpose representation learning architecture with state-of-the-art performance in tasks as diverse as equilibrium solving, deviation gain estimation and ranking.

In its most basic form, strategic interactions between players are formulated as NFGs where players simultaneously select actions and receive payoffs subject to the joint action. Strategic interactions are therefore presented as payoff tensors, with values to each player tabulated under every joint action. This tabular view of strategic interactions presents its own challenges to representation learning. Unlike modalities such as images and text whose spatial structure can be exploited for efficient representation learning \citep{lecun1995convolutional, hochreiter1997long}, the position of an action in the payoff tensor is unimportant: permuting the payoff matrix of any NFG yields an equivalent game --- an equivalency known as strongly isomorphic games \citep{McKinsey1951, gabarro2011complexity}. This inherent equivariance to NFGs has inspired prior works to compose order-invariant pooling functions in the neural network architecture for efficiency, albeit at the expense of the generality of the representation \citep{feng2021neural, marris2022turbocharging}.

We aim to leverage this inherent equivariance of NFGs while preserving full generality of the learned representation. This implies several desiderata that we discuss in turn. First, the representation needs to be versatile, allowing for inquiries at the level of individual actions, joint-actions, per-player or for the entire game. Second, it needs to be equivariant: for any per-action inquiry, the outputs for two actions should be exchanged if their positions are exchanged in the payoff tensor. Third, the embedding function should not assume that outcomes of all joint-actions are observed --- the representation should accommodate incomplete NFGs in a principled way. Fourth, the function should apply to games of different sizes. This implies that the number of network parameters should be independent from the size of the games \citep{hartford2016deep}. Finally, it would be desirable if the network architecture is interpretable, allowing for inspection at different stages of the network.

In the rest of this paper, we show how our proposed encoder architecture, NfgTransformer, satisfies all these desiderata simultaneously. The key idea behind our approach is to consider the embedding function class that represents an NFG as action embeddings, reversing the generative process from actions to payoffs. Action embeddings can be suitably composed to answer questions at different granularities and allows for equivariance in a straightforward way --- permutations of actions or players in the payoff tensor shall be reflected as a permutation in the action embeddings. We argue that NfgTransformer is a competitive candidate for general-purpose equivariant representation learning for NFGs, bridging the gap between deep learning and game-theoretic reasoning.

\section{Background}
\label{sec:background}

\paragraph{Normal-form Games}
NFGs are the fundamental game formalism where each player $p$ simultaneously plays one of its $T$ actions $a_p \in \{ a_p^1, \dots, a_p^T \} = \cA_p$ and receives a payoff $G_p: \cA \rightarrow \sR$
as a function of the joint action $a = (a_1, \dots, a_N) \in \cA$ of all $N$ players. Let $a = (a_p, a_\notp)$ with $a_\notp = (\dots, a_{p-1}, a_{p+1}, \dots) \in \cA_\notp$ the actions of all players except $p$. 
Let $\sigma(a) = \sigma(a_p, a_{\notp})$ denote the probability of players playing the joint action $a$ and $\sigma$ a probability distribution over the space of joint actions $\cA$. A pure strategy is an action distribution that is deterministic when a mixed-strategy can be stochastic. The value to player $p$ under $\sigma$ is given as $\expt_{a \sim \sigma}[G_p(a_p, a_\notp)]$. We refer to the payoff tensor tabulated according to the action and player ordering above as $G$.  %

\paragraph{Nash Equilibrium (NE)} Under a mixed joint strategy that factorises into an outer product of marginals $\sigma = \bigotimes_p \sigma_p$, player $p$'s unilateral deviation incentive is defined as
\begin{equation} \label{eq:delta_p}
    \delta_p(\sigma) = \max_{a'_p \in \cA_p} \expt_{a \sim \sigma}[G_p(a'_p, a_\notp) - G_p(a)].
\end{equation}
A factorisable mixed joint strategy $\sigma = \bigotimes_p \sigma_p$ is an $\epsilon$-NE if and only if $\delta(\sigma) = \max_p \delta_p(\sigma) \le \epsilon$. We refer to this quantity as the {\sc NE Gap} as it intuitively measures the distance from $\sigma$ to an NE of the game. A mixed-strategy NE is guaranteed to exist for a finite game \citep{nash1951} but exactly computing a normal-form NE beyond two-player zero-sum is PPAD-complete \citep{chen2009settling, daskalakis2009complexity}. If $\sigma$ is deterministic with $\sigma(a) = 1$, then $\delta(\sigma)$ or equivalently $\delta(a)$ defines the maximum deviation gain of the joint pure-strategy $a$. $a$ is a pure-strategy NE when $\delta(a) = 0$.

\paragraph{Permutation Equivariance}

Consider a strong isomorphism $\phi: G \rightarrow G'$ of NFGs \citep{gabarro2011complexity} with $\phi = ((\tau_p, p \in [N]), \omega)$, $\tau_p: a_p^i \rightarrow a_p^{i'}$ a player action permutation and $\omega: p \rightarrow p'$ a player permutation. Elements of the transformed game $G' = \phi(G)$ are therefore given as $$G'_{\omega(p)}\left(\tau_{\omega(1)}(a_{\omega(1)}), \dots, \tau_{\omega(N)}(a_{\omega(N)})\right) = G_p(a_1,...,a_N).$$

An encoder $f: G \rightarrow (\bA_1, \dots, \bA_N)$, with $\bA_p = (\va^1_p, \dots, \va^T_p)$ the action embeddings for player $p$, is said to be {\em equivariant} if
\begin{align}
\label{eq:perm_equivariance}
 f(\phi(G)) = (\bA'_1, \dots, \bA'_N) \text{ with }
 \bA'_{\omega(p)} = (\tau_{p}(\va^1_{p}), \dots, \tau_{p}(\va^T_{p}))
\end{align}

Here we slightly abuse the notation of $\tau_p$ to operate over action embeddings. Intuitively, permutation equivariance implies that $\phi$ and $f$ commute, or $f(\phi(G)) = \phi(f(G))$. We adopt the convention that the player permutation $\omega$ is applied after player action permutations $\tau_p, \forall p$.

\paragraph{Multi-Head Attention}

We describe self- and cross-attention QKV mechanisms that have become ubiquitous thanks to their generality and potential to scaling \citep{vaswani2017attention, dosovitskiy2020image, jaegle2021perceiver}. Both operations extend the basic QKV attention mechanism as follows:

\begin{equation}
\label{eq:attention}
\text{Attention}(Q, K, V) = \text{softmax}\left(\frac{QK^T}{\sqrt{d_k}}\right) V
\end{equation}

with $Q \in \bR^{n_q \times d_k}$, $K \in \bR^{n_k \times d_k}$ and $V \in \bR^{n_k \times d_v}$ the $n_q$ queries and $n_k$ key-value pairs. Conceptually, the attention mechanism outputs a weighted sum of $n_k$ values for each of $n_q$ query vectors, whose weights are determined by pairwise dot product between the key and query vectors. The output is of shape $\bR^{n_q \times d_v}$. The inputs QKV are outputs from fully-connected networks themselves, with $Q$ a function of $x_q \in \bR^{n_q \times d_{x_q}}$ and $K, V$ projected from the same input $x_{kv} \in \bR^{n_k \times d_{x_{kv}}}$. 
The attention operation is: 1) {\em order-invariant} with respect to $x_{kv}$; and 2) {\em equivariant} with respect to $x_q$. These are the key properties we leverage in the design of the NfgTransformer to achieve its permutation equivariance property. We refer to an attention layer as self-attention when $x_q$ is the same as $x_{kv}$, and cross-attention if not. In practice, each attention layer may have $H$ attention heads performing the attention operation of Equation~\ref{eq:attention} in parallel. This enables the attention layer to aggregate multiple streams of information in one forward pass. %

\section{Equivariant Game Representation}
\label{sec:equivariance}

While we have informally motivated the need for equivariant embedding functions, we formally state two practical implications of an equivariant embedding function that follow from a general theorem on the conditions under which two actions must have identical embeddings given an equivariant embedding function. For conciseness, we defer all formal statements and proofs to Appendix~\ref{app:equivariance}.

\begin{proposition}[Repeated Actions]
\label{prop:repeated}
    If $G(a^i_p, a_\notp) = G(a^j_p, a_\notp), \forall a_\notp$ and $f$ is deterministic and equivariant with $f(G) = (\dots, (\dots, \va^i_p, \dots, \va^j_p, \dots), \dots)$ then it follows that $\va^i_p = \va^j_p$.
\end{proposition}

\begin{proposition}[Player Symmetry]
\label{prop:symmetry}
    If player $p$ and $q$ are symmetric, $f$ is deterministic and equivariant with $f(G) = (\dots, \bA_p, \dots, \bA_q, \dots)$, then $\bA_p$ and $\bA_q$ are identical up to permutation.
\end{proposition}

Proposition~\ref{prop:repeated} guarantees by construction that repeated actions are treated identically in any downstream applications. Proposition~\ref{prop:symmetry} guarantees that player symmetry are reflected in the action embedding space which we show empirically in Section~\ref{sec:interpretation} for the NfgTransformer. 

\section{NfgTransformer}
\label{sec:method}

\begin{figure}
    \centering
    \includegraphics[width=\textwidth]{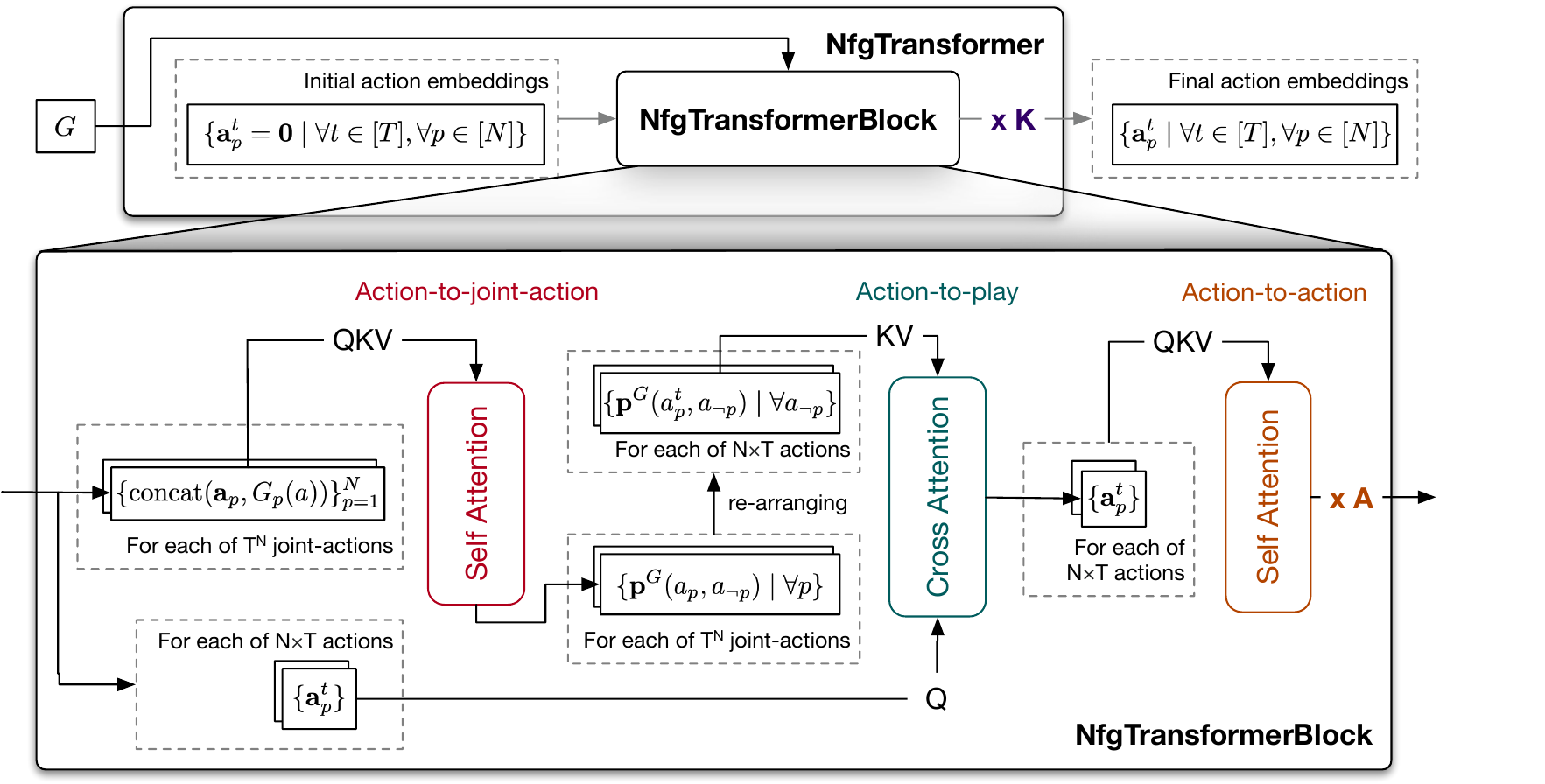}
    \caption{An overview of the NfgTransformer. The payoff tensor $G$ is encoded as action embeddings $\{\va^t_p \mid \forall t \in [T], \forall p \in [N]\}$ (Top). Action embeddings are zero-initialised and iteratively updated through a sequence of $K$ NfgTransformer blocks (Bottom). An arrow labeled with ``(Q)KV'' originates from a set of input (query-)key-values and terminates at a set of outputs. Each dashed box denotes an unordered set of elements of a specific type and cardinality.}
    \label{fig:nfg_transformer}
\end{figure}

We now describe the NfgTransformer, an encoder network that factorises a payoff tensor $G$ into action embeddings via a sequence of $K$ NfgTransformer blocks (Figure~\ref{fig:nfg_transformer} (Top)), each composed of a sequence of self- and cross-attention operations (Figure~\ref{fig:nfg_transformer} (Bottom)). We then show concrete examples of decoders and loss functions in Figure~\ref{fig:nfg_transformer_decoder} for several game-theoretic tasks at different decoding granularities, showing the generality of our approach. Finally, we discuss how the NfgTransformer naturally incorporates incomplete payoff tensors leveraging the flexibility of attention operations.

\subsection{Initialisation \& Iterative refinement}

Permutation equivariance implies that action embeddings must be agnostic to how players and actions are ordered in the payoff tensor. This suggests that action embeddings across players and actions must be initialised identically. We zero-initialise all action embeddings $\bA = \{\va^t_p = \mathbf{0} \mid \forall t \in [T], \forall p \in [N] \}$ with $\va^t_p \in \bR^D$. Upon initialisation, action embeddings are iteratively refined via a sequence of $K$ blocks given current action embeddings and the payoff tensor $G$. Each block returns updated action embeddings via self- and cross-attention operations that we describe in turn.

\paragraph{action-to-joint-action self-attention} represents a {\em play} of each action $a_p$ under a joint-action $a = (a_p, a_\notp)$ given payoff values to all players. Recall the definition of a self-attention operation, $x_q = x_{kv} = \{\text{concat}(\va_p, G_p(a))\}^N_{p=1}$, yielding one vector output per action for every joint-action. We refer to the output $\vp^G(a_p, a_\notp)$ as a {\em play} of $a_p$ under the joint action $a$ and payoffs $G$;

\paragraph{action-to-play cross-attention} then encodes information from all {\em play}s of each action $a^t_p$, with key-values $x_{kv} = \{ \vp^G(a^t_p, a_\notp) \mid \forall a_\notp \}$ and $x_q = \{ \va^t_p \}$, a singleton query. This operation yields a singleton output, $\{ \va^t_p \}$, as a function of all its {\em play}s and its input action embedding vector;

\paragraph{action-to-action self-attention} then represents each action given all action embeddings. Here, $x_q = x_{kv} = \{ \va^t_p \mid \forall p \in [N], \forall t \in [T] \}$. We ablate this operation (by varying $A$) in Section~\ref{sec:l2ball}, showing its benefits in propagating information across action embeddings.

Within each block, equivariance is preserved given key-value order-invariance, and query equivariance properties of the attention operation. Each output embedding $\va^t_p$ is a function of its own embedding at input, its unordered set of {\em play}s, and the unordered set of all action embeddings.

\subsection{Task-specific decoding}

The resulting action embeddings can be used for a variety of downstream tasks at different decoding granularities. We describe and empirically demonstrate three use-cases in specifics (Figure~\ref{fig:nfg_transformer_decoder}).
\paragraph{Nash equilibrium-solving} requires decoding at the level of each action, estimating a marginal action distribution for each player $\hat\sigma_p = \text{softmax}(w^1_p, \dots, w^T_p)$ where $w^t_p = \text{MLP}(\va^t_p)$ is the logit for an action $a^t_p$. Here, we follow \citet{duan2023nash} in minimising the loss function $\max_p \delta_p(\hat\sigma)$ end-to-end via gradient descent with $\hat\sigma = \otimes_p \hat\delta_p$ and $\delta_p(\hat\sigma) = \max_{a'_p \in \cA_p} \expt_{a \sim \hat{\sigma}}[G_p(a'_p, a_\notp) - G_p(a)]$.
\paragraph{Max deviation-gain estimation} decodes a scalar estimate for each joint-action $a = (a_1, \dots, a_N), \forall a \in \cA$. Here, we represent each joint action as $a = \sum_p \va_p$ and estimate its maximum deviation gain $\hat\delta(a)$ by minimising $L_2(\hat\delta(a), \max_p \delta_p(a)) = (\hat\delta(a) - \max_p \delta_p(a))^2$, with $\delta_p(a) = \max_{a'_p \in \cA_p} [G_p(a'_p, a_\notp) - G_p(a)]$ the deviation gain to player $p$ under the joint action $a$.
\paragraph{Payoff reconstruction} decodes a scalar for each payoff value $G_p(a), \forall p \in [N], \forall a \in \cA$. Here we use a self-attention operation for decoding, similar to the action-to-joint-action self-attention operation in the encoder but {\em without} appending to action embeddings their payoff values, which are to be reconstructed. To compute a reconstruction loss, we minimise $L_2(\hat{G}_p(a), G_p(a))$, with $\hat{G}_p(a)$ a function of the action embedding $\va_p$ and the unordered set of co-player action embeddings $\va_\notp$.

\begin{figure}
    \centering
    \includegraphics[width=\textwidth]{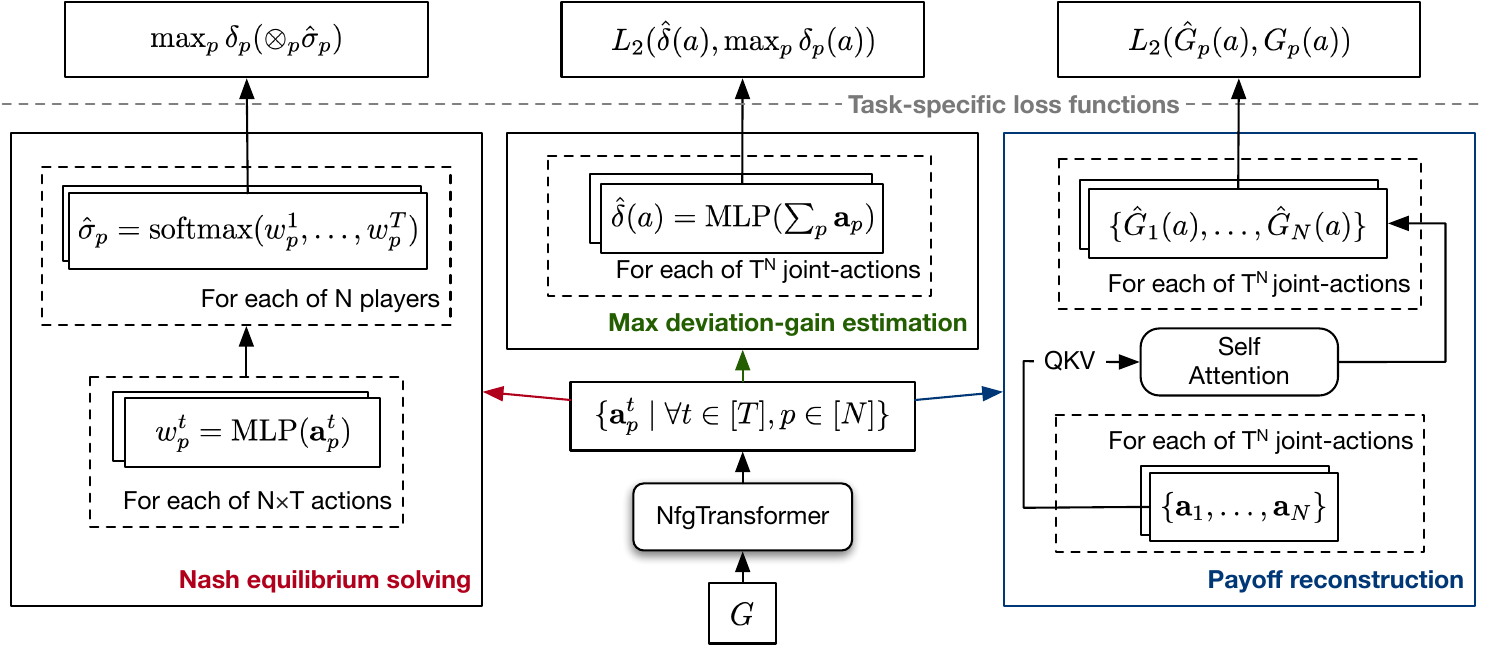}
    \caption{Example task-specific decoders and losses from general-purpose action embeddings.}
    \label{fig:nfg_transformer_decoder}
\end{figure}

\subsection{Representing incomplete games}
\label{sec:masking}

We have assumed thus far that outcomes for every joint-action of the game are observed. This is not always the case in practice as the costs of evaluating every joint-action can be prohibitive. Instead, one may infer outcomes for all joint-actions, given an incomplete NFG \citep{elo2008rating}. With a slight modification, the NfgTransformer accommodates such use-cases in a principled way. To do so, we extend the vanilla attention implementation (Equation~\ref{eq:attention}) to allow for additional binary mask vectors $\vm_q \in \{0, 1\}^{n_q}$ and $\vm_{kv} \in \{0, 1\}^{n_{k}}$ for the query and key-value inputs, indicating their validity. Equation~\ref{eq:attention_mask} defines this masked attention operation, with $\mathds{1}_\infty(1) = 0$ and $\mathds{1}_\infty(0) = \infty$.

\begin{equation}
\label{eq:attention_mask}
\text{MaskedAttention}(Q, K, V, \vm_q, \vm_{kv}) = \text{softmax}\left(\frac{QK^T - \mathds{1}_\infty(\vm_q \otimes \vm_{kv})}{\sqrt{d_k}}\right) V
\end{equation}

To represent incomplete NFGs, we use the Equation~\ref{eq:attention_mask} in lieu of Equation~\ref{eq:attention} in all self- and cross-attention operations and set masking vectors accordingly. For instance, if a joint-action is unobserved, then one would set $\vm_{kv}$ to reflect the validity of each {\em play} for the action-to-play operation. The NfgTransformer architecture is highly effective for representing actions in incomplete games and predict payoffs for unobserved joint-actions, as we show empirically in Section~\ref{sec:payoff_prediction}.

\section{Results}
\label{sec:results}

The goal of our empirical studies is three-fold. First, we compare the NfgTransformer to baseline architectures to demonstrate improved performance in diverse downstream tasks on synthetic and standard games. Second, we vary the model hyper-parameters and observe how they affect performance. We show in particular that some tasks require larger action embedding sizes while others benefit from more rounds of iterative refinement. Lastly, we study {\em how} the model learned to solve certain tasks by interpreting the sequence of learned attention masks in a controlled setting. Our results reveal that the solution found by the model reflects elements of intuitive solutions to NE-solving in games. For completeness, we discuss additional empirical results in Appendix~\ref{app:two_by_two}. For instance, we show all nontrivial equilibrium-invariant 2×2 games are embedded following structure of the known embedding space proposed in \citet{marris2023equilibrium}.

\subsection{Synthetic games}

\label{sec:l2ball}

\begin{table}[ht]
\caption{We compare NfgTransformer to baseline architectures in synthetic games. Each configuration is averaged across 5 independent runs. For NfgTransformer variants (Ours), we annotate each variant with corresponding hyper-parameters ($K$, $A$ and $D$ as shown in Figure~\ref{fig:nfg_transformer})\label{tab:synthetic_games}. We provide training curves with confidence intervals and parameter counts of each configuration in Appendix~\ref{app:l2_ball_games}.}
\centering
\tabcolsep=5pt\relax
\begin{small}
\begin{tabular}{c|cccc||cccc}
\hline
\multirow{3}{*}{Model} & \multicolumn{4}{c||}{NE (NE Gap)} & \multicolumn{4}{c}{Max-Deviation-Gain (MSE)}  \\
                         \cline{2-9}
                       & N=2 & N=2 & N=3 & N=3 & N=2 & N=2 & N=3 & N=3 \\
                       & T=16 & T=64 & T=8 & T=16 & T=16 & T=64 & T=8 & T=16 \\ \hline \hline
{\tt Ours(D=~32,K=2,A=1)}  & 0.2239 & 0.1685 & 0.1344 & 0.0796 & 0.0949 & 0.5008 & 0.5206 & 0.6649 \\
{\tt Ours(D=~32,K=4,A=1)}  & 0.0466 & 0.1096 & 0.0892 & 0.0553 & 0.0248 & 0.3755 & 0.3679 & 0.6173 \\
{\tt Ours(D=~32,K=8,A=1)}  & 0.0344 & 0.0554 & 0.0484 & 0.0334 & 0.0067 & 0.2989 & 0.3582 & 0.5825 \\
{\tt Ours(D=~64,K=8,A=0)}  & 0.0332 & 0.0661 & 0.0636 & 0.0384 & 0.0056 & 0.1549 & 0.2848 & 0.5583 \\
{\tt Ours(D=~64,K=8,A=1)}  & 0.0308 & 0.0545 & 0.0478 & 0.0325 & 0.0007 & 0.0784 & 0.1830 & 0.4961 \\
{\tt Ours(D=~64,K=8,A=2)}  & {\bf 0.0243} & 0.0542 & 0.0437 & 0.0314 & 0.0005 & 0.0759 & 0.1942 & 0.4725 \\
{\tt Ours(D=128,K=2,A=1)}  & 0.2090 & 0.1665 & 0.1274 & 0.0769 & 0.0159 & 0.1922 & 0.3154 & 0.5357 \\
{\tt Ours(D=128,K=4,A=1)}  & 0.0429 & 0.0981 & 0.0804 & 0.0530 & 0.0013 & 0.1361 & 0.0955 & 0.4153 \\
{\tt Ours(D=128,K=8,A=1)}  & 0.0308 & {\bf 0.0502} & {\bf 0.0412} & {\bf 0.0297} & {\bf 0.0001} & {\bf 0.0161} & {\bf 0.0487} & {\bf 0.3641} \\ \hline \hline
{\tt EquivariantMLP}  & 0.2770 & 0.2132 & {\bf 0.1431} & {\bf 0.0929} & 0.1789 & 0.8153 & 0.5433 & 0.7914 \\
{\tt MLP}  & 0.3905 & 0.3248 & 0.1741 & 0.1381 & 0.3854 & 0.8354 & 0.5623 & 0.7906 \\
{\tt NES}  & {\bf 0.0829} & {\bf 0.1635} & 0.1478 & 0.1140 & {\bf 0.0488} & {\bf 0.4860} & {\bf 0.4047} & {\bf 0.6480} \\ \hline
\end{tabular}

\end{small}
\end{table}

We first evaluate variations of the NfgTransformer architecture on synthetic games of varying sizes on NE equilibrium-solving and deviation gain estimation. To generate synthetic games with broad coverage, we follow \citet{marris2022turbocharging} which samples games from the equilibrium-invariant subspace, covering all strategic interactions that can affect the equilibrium solution of an NFG. Each game's payoff tensor $G$ has zero-mean over other player strategies and Frobenius norm $\left\| G_p \right\|_F = \sqrt{T^N}$. We compare our results to baseline MLP networks with numbers of parameters {\em at least} that of our largest transformer variant (at 4.95M parameters), an equivariant MLP network that re-arranges actions in descending order of their average payoffs, as well as an NES \citep{marris2022turbocharging} network that is designed for equilibrium-solving. See Appendix~\ref{app:l2_ball_games} for details on game sampling, network architectures and parameter counts of each model. We note that the parameter count of the NfgTransformer is independent of the game size, a desideratum of \citet{hartford2016deep}.

\subsubsection{Solving for NE equilibrium} 
\label{sec:ne_solving}
For equilibrium solving, we optimise variants of the NfgTransformer to minimise the {\sc NE Gap} $\delta(\hat\sigma) = \max_p \delta_p(\hat\sigma)$ (Figure~\ref{fig:nfg_transformer_decoder} (Left)). Table~\ref{tab:synthetic_games} (Left) shows our results. EquivariantMLP outperforms MLP \citep{duan2023nash}, demonstrating the importance of leveraging equivariance inherent to NFGs but remains ineffective at solving this task. NES \citep{marris2022turbocharging}, equivariant by construction, significantly outperforms both MLP variants in 2-player settings but trails behind in 3-player games. The NfgTransformer is also equivariant by construction but learns to capture relevant information without handcrafted payoff feature vectors. All NfgTransformer variants, most at fewer parameter count than baselines, significantly outperform across game sizes with near-zero {\sc NE Gap}.

Among the NfgTransformer variants, our results show a clear trend: increasing the number of transformer blocks (with $K \in [2, \dots, 8]$) improves performance, especially as the game becomes large. This makes intuitive sense, as it adds to the number of times the action embeddings can be usefully refined --- action embeddings at the end of one iteration become more relevant queries for the next. In contrast, the benefit of increased action embedding size is muted (with $D \in [32, \dots, 128]$). We hypothesise that for equilibrium-solving, information related to a subset of the available actions can often be ignored through iterative refinement (e.g. dominated actions), as they do not contribute to the final equilibrium solution. Lastly, we evaluate an NfgTransformer variant that does {\em not} perform any action-to-action self-attention ($A$ = 0). In this case, action embeddings for the {\em same} player do not interact within the same block and its performance is markedly worse. Of particular interests is the comparison between the variants with $A = 1$ and $A = 2$ where $A = 2$ demonstrates a benefits in the most complex games of size $16 \times 16 \times 16$ but not in smaller games. This suggests that action-to-action self-attention facilitates learning, especially in tasks that require iterative reasoning.

\subsubsection{Estimating Maximum Deviation Gains}

A related task is to determine what is the maximum incentive for any player to deviate from a joint pure-strategy $\sigma$ (or equivalently, a joint-action). This quantity is informative on the stability of a joint behaviour --- in particular, if a joint pure-strategy has a maximum deviation gain $\delta(a)$ of zero, then by definition we have found a pure-strategy NE. We optimise a NfgTransformer network to regress towards the maximum deviation-gain $\delta(a)$ for every joint pure-strategy $a$, using a per joint-action decoder architecture (Figure~\ref{fig:nfg_transformer_decoder} (Middle)). We report the regression loss in mean squared error of different architecture variants in Table~\ref{tab:synthetic_games} (Right). We observe that NES consistently outperforms MLP baselines, but underperforms the NfgTransformer variants as the size of the game increases.

Similar to our observations in Section~\ref{sec:ne_solving}, the number of transformer blocks played a role in transformer variants' final performance. However, it is no longer the main factor. Instead, the action-embedding size $D$ becomes critical. Variants with higher embedding size $D = 128$ can be competitive, even for {\em shallow} models (e.g. $K = 4$ in $16 \times 16 \times 16$ games). This can be explained by the lack of structure in the underlying game, as payoff tensors cover the full equilibrium-invariant subspace of NFGs: payoffs of one joint-action does not provide any information on the outcomes of another. To perform well, the model must learn to memorise outcomes of different joint-actions a reduced action embedding size can become a bottleneck.

\subsection{Payoff Prediction in Empirical \textbf{{\sc Disc}} games}

\label{sec:payoff_prediction} 

\begin{figure*}
    \centering
    \includegraphics[width=\textwidth]{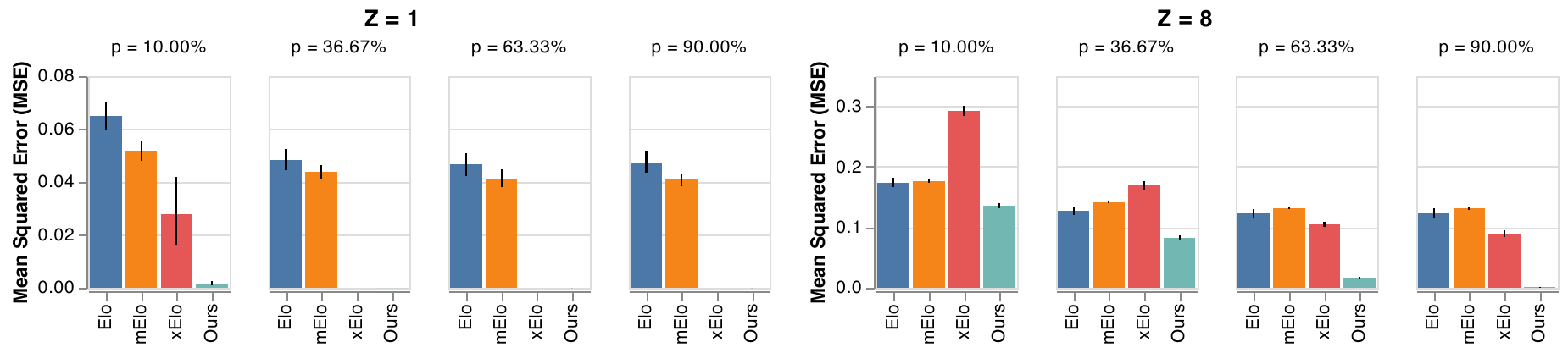}
    \caption{Payoff prediction error averaged over all players across {\em unobserved} joint-actions. Results are averaged over 32 randomly sampled empirical \textsc{DISC} games in each game configuration.}
    \label{fig:disc_games}
\end{figure*}

What if the game class follows a {\em structured} generative process? This is often the case in practice and a useful representation learning method should capture any such latent structure. We turn to {\sc Disc} games \citep{balduzzi2019open} to evaluate the efficacy of NfgTransformer in this case, compared to several payoff prediction methods from the ranking literature. 

\begin{definition}[{\sc Disc} Game]
\label{def:disc}
Let $\vu_t, \vv_t \in \sR^{Z}, t \in [T]$, the win-probability to action $i$ when playing against $j$ is defined as $P_{ij} = \sigma( \vu_i^T \vv_j - \vu_j^T \vv_i) = 1 - P_{ji}$ with $\sigma(x) = \frac{1}{1 + e^{-x}}$ the sigmoid function.
\end{definition}

Definition~\ref{def:disc} describes a class of symmetric two-player zero-sum games where the outcomes are defined by latent vectors $\vu_t, \vv_t \in \sR^{Z}, t \in [T]$, generalising the original definition of {\sc Disc} game \citep{balduzzi2019open} to allow for latent vectors with $Z > 1$. Payoff values under one joint-action therefore become informative for predicting the outcomes of others. The amount of information that can be inferred about one joint-action from knowing another is controlled by the latent vector dimension $Z$. The problem setting is as follows. Each algorithm is given access to an empirical game where outcomes for each joint-action is observed with probability $p$. The goal is to accurately predict the outcome to each player under {\em unobserved} joint-actions.

A rich body of literature have been dedicated to solving this task owing to its relevance in real-world competitions. In Go and Chess, players are classically assigned Elo ratings which are designed to be predictive of the win-probability between any two ranked players. Several improvements to Elo have been proposed since, recognising the many limitations of Elo. We compare our results to methods such as Elo \citep{elo2008rating}, mElo \citep{balduzzi2018re} and xElo \citep{bertrand2023limitations} across different settings of the {\sc Disc} game. Figure~\ref{fig:disc_games} shows our results in MSE averaged across {\em unobserved} joint-actions. 
NfgTransformer outperforms all baselines significantly across all settings. In particular, NfgTransformer recovered the latent variable game (i.e. $Z = 1$) near perfectly as soon as 10\% of the joint-actions are observed, with an error rate an order of magnitude lower than the second best method. This result is particularly remarkable as baseline methods are designed with {\sc Disc} games in mind, when NfgTransformer is not. At $Z = 8$, NfgTransformer continues to outperform, with its prediction accuracy degrading gracefully as fewer joint actions are observed. Our results suggest that NfgTransformer is highly effectively at recognising and exploiting the latent structure in games if it exists. We provide details on game sampling, masking, network architecture and baseline implementation in Appendix~\ref{app:disc_ranking}. For NfgTransformer, outcomes of the unobserved joint-actions are masked out following the procedure described in Section~\ref{sec:masking}. At training time, the model minimises reconstruction loss for {\em all} joint-actions.

\subsection{Interpreting NfgTransformer}
\label{sec:interpretation}

\begin{figure}
    \centering
    \includegraphics[width=\textwidth]{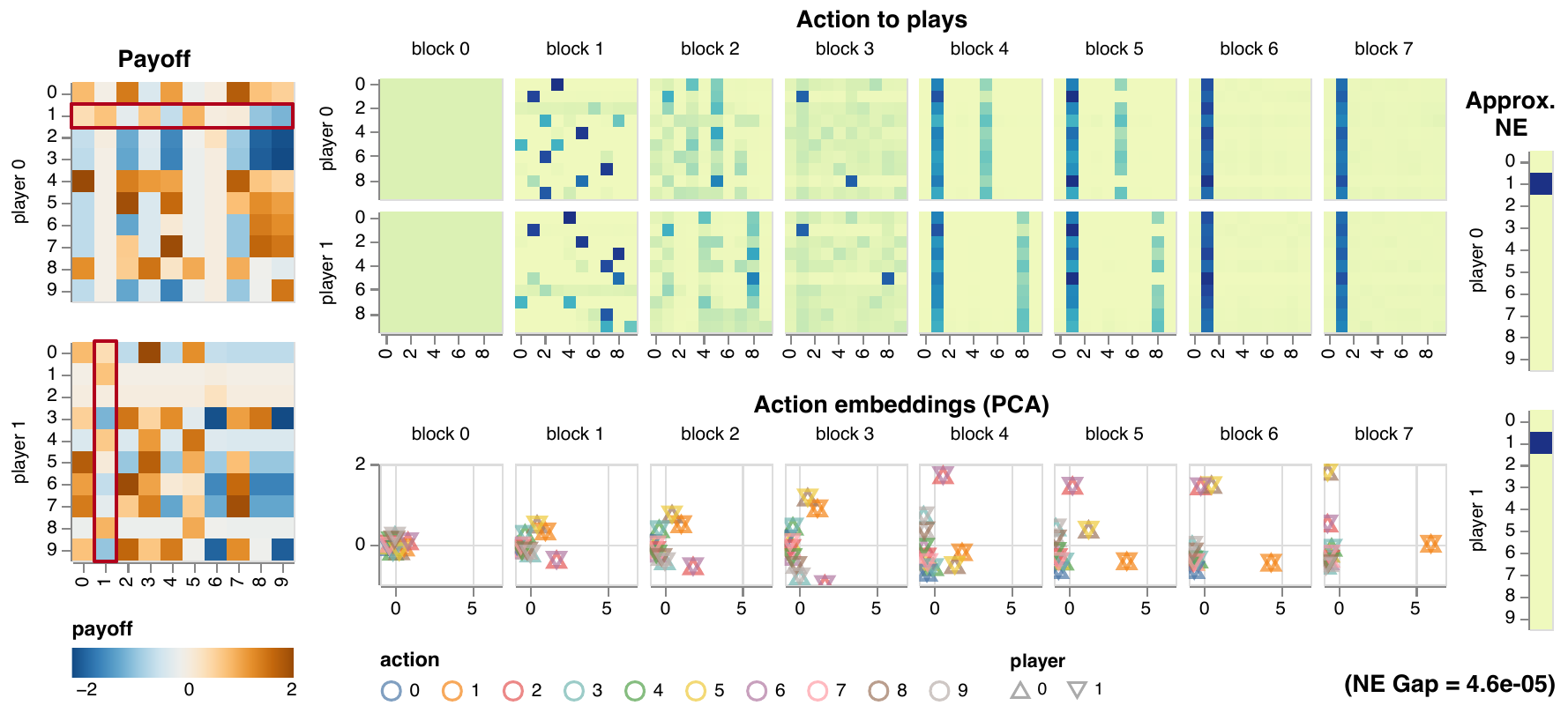}
    \caption{Visualisation of attention masks and action-embeddings at inference time on a held-out instance of Bertrand Oligopoly game whose payoff tensor is shown on the Left and the inferred NE strategy profile shown on the right (with {\sc NE Gap} at near zero). The equilibrium pure-strategies for the two players are shown in red. The sequence of 8 {\em action-to-play} attention masks and the PCA-reduced action-embeddings at the end of each transformer block are shown in the middle.}
    \label{fig:gamut}
\end{figure}

NfgTransformer explicitly reasons about action-embedding through structured attention mechanism, allowing for inspection at each stage of the iterative refinement process. We exploit this property and investigate {\em how} the model implements game-theoretic algorithms such as equilibrium-solving in a suite of 22 GAMUT games representing diverse strategic interactions \citep{nudelman2004run, porter2008simple} that are classically studied in the economics and game theory literature. We optimised an NfgTransformer network as in Section~\ref{sec:l2ball}, but focused on $10 \times 10$ games and removed the action-to-action self-attention (i.e. $A = 0$) for a more concise visualisation. Instead of attending to multiple pieces of information in parallel, each attention layer is implemented by a single attention head ($H=1$). We show in Appendix~\ref{app:gamut} that this simplified architecture is sufficient in approximating NE in games of this size to reasonable accuracy. We visualise the network on held-out games.

Figure~\ref{fig:gamut} illustrates an instance of iterative NE solving in a game of Bertrand Oligopoly in action. The NfgTransformer successfully identified a pure-strategy NE (Right) with near zero {\sc NE Gap} (Top-Right). Inspecting the payoff tensor (Left), we verify that player 0 playing action 2 and player 1 countering with action 1 is indeed a pure-strategy NE. Is this solution reflected in the underlying attention mechanism? Indeed, the action-to-plays attention masks (Top) appear to have converged to a recurring pattern where each action attends to the equilibrium-strategy of its co-player. This is remarkable for two reasons. First, NE by definition implies indifference across action choices --- when co-players implement the equilibrium strategy, the focal player should be indifferent to its action choice. Here, this indifference played out. Second, an equilibrium solution should be stable which appears to be the case over the last few iterations of action-to-plays attentions. Zooming in on the attention masks at earlier iterations, we see that following zero-initialisation of all action embeddings, the attention mask for each action is equally spread across all its {\em plays}. The attention masks at the next block, by contrast, appear structured. Indeed, each action attends to the {\em play} that involves the best-response from the co-player. It is worth noting that this pattern of attending to one's best-response in NE-solving emerged through learning, without prior knowledge. 

The action embeddings themselves also reveal interesting properties of the game, such as symmetry (Proposition~\ref{prop:symmetry}). While it might not be immediately obvious that the payoff matrices in Figure~\ref{fig:gamut} can be made symmetric by permuting actions, the action embeddings (Bottom) show that action-embeddings across the two players overlap exactly, thanks to the inherent equivariance of the NfgTransformer. The ordering of players and actions is unimportant --- the representation of an action is entirely driven by its outcomes under the joint-actions. Here, action embeddings revealed the inherent symmetry of Bertrand Oligopoly \citep{bertrand1883theorie}.

For completeness, we offer similar analysis of the NfgTransformer applied to other game classes in Appendix~\ref{app:gamut}, including games without symmetry, with mixed-strategy NE as well as instances where NfgTransformer failed. Additionally, we show examples of applying NfgTransformer to out-of-distribution games such as Colonel Blotto \citep{roberson2006colonel} that presents several strategy cycles.

\section{Related Works}

Recent works leveraged deep learning techniques to accelerate scientific discoveries in mathematics \citep{fawzi2022discovering}, physics \citep{PhysRevResearch.2.033429} and biology \citep{jumper2021highly}. Our work is similarly motivated, but brings deep learning techniques to game theory and economic studies. We follow a line of works in bringing scalable solutions to game theory \citep{marris2022turbocharging, duan2023nash, vadori2023ordinal}, or integrating components for strategic reasoning as a part of a machine learning system \citep{hu2003nash, greenwald2003correlated, feng2021neural, liu2022neupl, liu2022simplex}.
In game theory, \citet{hartford2016deep} is the closest to our work and the first to apply deep representation learning techniques to cognitive modelling of human strategic reasoning using NFGs. The authors systematically outlined a number of desiderata for representation learning of two-player NFGs, including player and action equivariance as well as the independence between the number of parameters of the learned model and the size of the game. The NfgTransformer satisfies these desiderata but applies to $n$-player general-sum games. \citet{wiedenbeck2023data} studies efficient data structures for payoff representation; our approach can be readily integrated into deep learning systems without any assumption on the games.

\section{Conclusion}

We proposed NfgTransformer as a general-purpose, equivariant architecture that represents NFGs as action embeddings. We demonstrate its versatility and effectiveness in a number of benchmark tasks from different sub-fields of game theory literature, including equilibrium solving, deviation gain estimation and ranking. We report empirical results that significantly improve upon state-of-the-art baseline methods, using a unified representation learning approach. We show that the resulting model is also interpretable and parameter-efficient. Our work paves the way for integrating game-theoretic reasoning into deep learning systems as they are deployed in the real-world.

\subsubsection*{Acknowledgments}
We are grateful to Bernardino Romera-Paredes for the productive discussion on the different considerations in designing an equivariant neural architecture, to Wojciech M. Czarnecki for his expertise in ranking and evaluation in games and to Skanda Koppula for his advice on optimisation techniques for the transformer architecture.

\bibliography{iclr2024_conference}
\bibliographystyle{iclr2024_conference}

\newpage
\appendix

\section{Permutation Equivariant Representation of NFGs}
\label{app:equivariance}

We provide formal statements on identities in action embedding representation when using a deterministic, permutation equivariant encoder for NFGs. First, we recall the definition of a strong isomorphism between two NFGs \citep{McKinsey1951, gabarro2011complexity}.

\begin{definition}[Strongly Isomorphic Games]
    \label{def:strong_isomorphism}
    Let $G$ and $G'$ be two NFGs. $G$ and $G'$ are said to be strongly isomorphic and $\phi$ a strong isomorphism if $\phi = ((\tau_p, p \in [N]), \omega)$ with $\tau_p: a_p^i \rightarrow a_p^{i'}$ a player action permutation and $\omega: p \rightarrow p'$ a player permutation such that $G'_{\omega(p)}\left(\tau_{\omega(1)}(a_{\omega(1)}), \dots, \tau_{\omega(N)}(a_{\omega(N)})\right) = G_p(a_1,...,a_N), \forall a \in \cA$.
\end{definition}

To make Definition~\ref{def:strong_isomorphism} concrete, consider the coordination and anti-coordination games shown in Figure~\ref{fig:isomorphic_games}. The two games are strongly isomorphic because there exists a strong isomorphism $\phi = ((\tau_1 = (a^1_1 \to a^2_1, a^2_1 \to a^1_1), \tau_2 = (a^1_2 \to a^1_2, a^2_2 \to a^2_2)), \omega = (1 \to 1, 2 \to 2))$. As aside, \citet{McKinsey1951} calls strongly isomorphic games strategically equivalent which we discuss soon.

\begin{figure}[ht]
\centering
\begin{subfigure}[t]{0.49\linewidth}
    \centering
    \payoffstable[%
        label=$G$,
        row player first strategy label=$a_1^1$,%
        row player second strategy label=$a_1^2$,%
        column player first strategy label=$a_2^1$,%
        column player second strategy label=$a_2^2$,%
    ]{1}{0}{0}{1}{1}{0}{0}{1}
    \caption{Coordination game\label{fig:coord_game}}
\end{subfigure}
\begin{subfigure}[t]{0.49\linewidth}
    \centering
    \payoffstable[%
        label=$G'$,
        row player first strategy label=$a_1^1$,%
        row player second strategy label=$a_1^2$,%
        column player first strategy label=$a_2^1$,%
        column player second strategy label=$a_2^2$,%
    ]{0}{1}{1}{0}{0}{1}{1}{0}
    \caption{Anti-coordination game\label{fig:anticoord_game}}
\end{subfigure}
\caption{Strongly isomorphic games.\label{fig:isomorphic_games}}
\end{figure}

We additionally note that in the special case when $G$ is $G'$, $\phi$ is referred to as a strong \emph{auto}morphism.

\begin{definition}[Strongly Automorphic Game]
    \label{def:strong_automorhpism}
    $G$ is said to be strongly automorphic and $\phi$ a strong automorphism if $\phi = ((\tau_p, p \in [N]), \omega)$ with $\tau_p: a_p^i \rightarrow a_p^{i'}$ a player action permutation and $\omega: p \rightarrow p'$ a player permutation such that $G_{\omega(p)}\left(\tau_{\omega(1)}(a_{\omega(1)}), \dots, \tau_{\omega(N)}(a_{\omega(N)})\right) = G_p(a_1,...,a_N), \forall a \in \cA$.
\end{definition}

For instance, the coordination game (Figure~\ref{fig:coord_game}) is also strongly automorphic as there exists three automorphisms that recover the same game.
\begin{subequations}
\begin{align}
    \phi &= ((\tau_1 = (a^1_1 \to a^2_1, a^2_1 \to a^1_1), \tau_2 = (a^1_2 \to a^2_2, a^2_2 \to a^1_2)), \omega = (1 \to 1, 2 \to 2))  \label{eq:coord_automorhpism_a} \\
    \phi &= ((\tau_1 = (a^1_1 \to a^2_1, a^2_1 \to a^1_1), \tau_2 = (a^1_2 \to a^2_2, a^2_2 \to a^1_2)), \omega = (1 \to 2, 2 \to 1)) \\
    \phi &= ((\tau_1 = (a^1_1 \to a^1_1, a^2_1 \to a^2_1), \tau_2 = (a^1_2 \to a^1_2, a^2_2 \to a^2_2)), \omega = (1 \to 2, 2 \to 1))  \label{eq:coord_automorhpism_c}
\end{align}
\end{subequations}

Note that $\phi$ is a permutation over all the players' actions which is a composition of the player and action permutations $\phi = \tau_1 \cdot ... \cdot \tau_N \cdot \omega$. Therefore $\phi$ is not a general permutation, but a structured one. We use a convention that the player permutation is applied last. Finally, we recall any permutation $\pi$ can be written uniquely as $m$ permutation orbits with $\pi = C^1,\dots,C^m$, each operating on a disjoint (possibly singleton) subset of elements that $\pi$ operates over. Therefore $\phi$ is also a collection of permutation orbits.

Considering the coordination game again, the automorphisms (Equations~\ref{eq:coord_automorhpism_a}-\ref{eq:coord_automorhpism_c}) can be written as permutations which each consists of two orbits containing two actions each.

\begin{subequations}
\noindent\begin{minipage}{0.3\linewidth}
    \begin{align} \label{eq:perma}
        \drawpermutation{B}{A}{D}{C}{north}{south}{north}{south}
    \end{align}
\end{minipage} \hfill
\noindent\begin{minipage}{0.3\linewidth}
    \begin{align} \label{eq:permb}
        \drawpermutation{D}{C}{B}{A}{north}{north}{south}{south}
    \end{align}
\end{minipage} \hfill
\noindent\begin{minipage}{0.3\linewidth}
    \begin{align} \label{eq:permc}
        \drawpermutation{C}{D}{A}{B}{north}{north}{south}{south}
    \end{align}
\end{minipage}
\end{subequations}

~

\begin{definition}[Strategically Equivalent Actions]
    \label{def:strageic_equivalence}
    Two actions $a_p^i$ and $a_q^j$ are strategically equivalent if there exists a strong automorphism, $\phi$ which contains $a_p^i$ and $a_q^j$ in an orbit. Equivalently, two actions $a_p^i$ and $a_q^j$ are strategically equivalent if there exists a strong automorphism, $\phi = ((..., \tau_p=(..., i\to j, ...)), \omega=(..., p\to q))$.
\end{definition}

Again, consider the running example of the coordination game. From Equation~\ref{eq:perma} we can see that $(a_1^1, a_1^2)$, and $(a_2^1, a_2^2)$ are each strategically equivalent pairs. Furthermore, from Equation~\ref{eq:permb} we can see that $(a_1^1, a_2^2)$, and $(a_1^2, a_2^1)$ are also each strategically equivalent pairs. Therefore in the coordination games all the actions are strategically equivalent to each other.

\begin{theorem}
    If an embedding function, $f$, is deterministic and equivariant, then strategically equivalent actions, $\va^i_p$ and $\va^j_q$, must have the same embeddings. 
\end{theorem}
\begin{proof}
    The embedding function, $f$, is deterministic and equivariant over players and actions. Additionally, if $\phi$ is an automorphism of $G$, then $f(G) = f(\phi(G)) = \phi(f(G))$. Therefore the embeddings are also equal, $\va^i_p = \va^j_q$. 
\end{proof}

\begin{proposition}[Repeated Actions]
    If $G(a^i_p, a_\notp) = G(a^j_p, a_\notp), \forall a_\notp$ and $f$ is deterministic and equivariant with $f(G) = (\dots, (\dots, \va^i_p, \dots, \va^j_p, \dots), \dots)$ then it follows that $\va^i_p = \va^j_p$.
\end{proposition}
\begin{proof}[Proof]
    If actions are repeated, there there exists an automorphism $\phi=((..., \tau_p=(..., i \to j, ...), ...), \omega=identity)$. Therefore $a_p^i$ and $a_q^j$ are strategically equivalent and have the same embeddings, $\va^i_p = \va^j_p$.
\end{proof}

\begin{proposition}[Player Symmetry]
    If player $p$ and $q$ are symmetric, $f$ is deterministic and equivariant with $f(G) = (\dots, \bA_p, \dots, \bA_q, \dots)$, then $\bA_p$ and $\bA_q$ are identical up to permutation.
\end{proposition}
\begin{proof}[Proof]
    If the game is symmetric between $p$ and $q$, there there exists an automorphism $\phi=((...,\tau_p,...), \omega=(...,p \to q, ...))$. Therefore $\tau_p(a_p^i)$ and $a_q^i$ are strategically equivalent for all $i$, and have the same embeddings, $\tau_p(\bA_p) = \bA_q$.
\end{proof}

\section{Experimental Setup}

\subsection{Supervised Learning in Synthetic Games}
\label{app:l2_ball_games}

\begin{figure*}
    \centering
    \includegraphics[width=\textwidth]{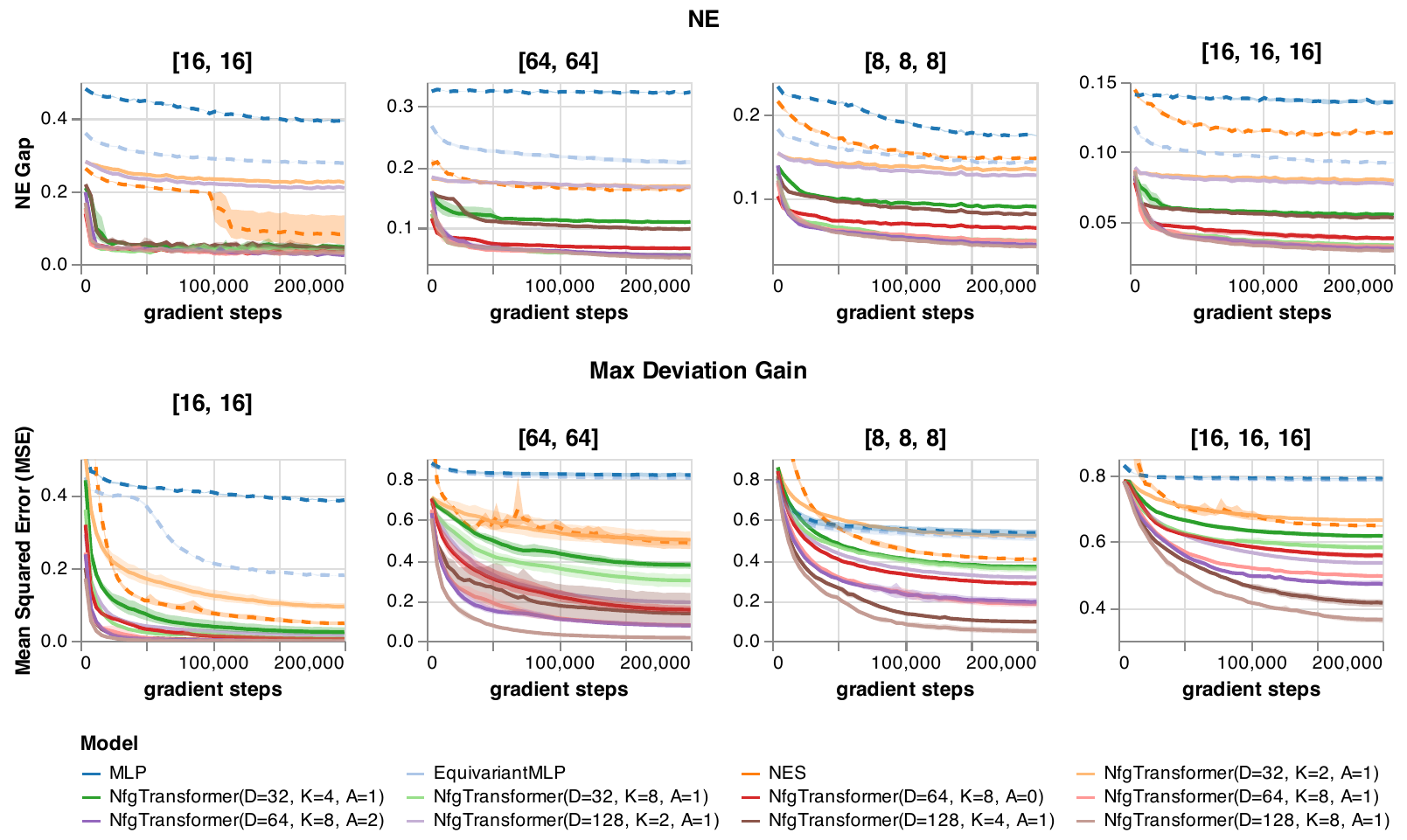}
    \caption{We compare NfgTransformer to baseline architectures in synthetic games. Results from baseline experiments are shown in dashed lines. Each configuration is averaged across 5 independent runs with shaded areas representing the confidence intervals. For NfgTransformer variants, we annotate each variant with corresponding hyper-parameters ($K$, $A$ and $D$ as shown in Figure~\ref{fig:nfg_transformer}).}
    \label{fig:synthetic_games}
\end{figure*}

\paragraph{Games Sampling}
Games are sampled from the equilibrium-invariant subspace \citep{marris2023equilibrium,marris2022turbocharging}, with zero-mean payoff over other player actions and a unit variance ($\sqrt{T^N}$) Frobenius tensor norm over player payoffs. To sample uniformly over such a set, first sample a game from a unit normal distribution, $G_p \sim \mathcal{N}(0, 1)$, and then normalize.
\begin{align}
    G^\text{equil}_p(a) = \frac{\sqrt{T^N}}{Z} \left(G_p(a) - \frac{1}{T} \sum_{a_p} G_p(a_p, a_{-p}) \right) %
    \quad %
    Z = \left\| G_p - \frac{1}{T} \sum_{a_p} G_p(a_p, a_{-p}) \right\|_F %
\end{align}
The benefit of this distribution is that it provides a way to uniformly sample over the space of all possible strategic interactions in a NFG of a specific shape. The equilibrium-invariant subspace has lower degree of freedom than a full NFG, freeing the neural network from having to learn offset and scale invariance. Any game can be simply mapped to the equilibrium-invariant subspace without changing its set of equilibria.

\paragraph{Architecture}

We provide additional technical details on the network architectures presented in Section~\ref{sec:l2ball}. The baseline MLP networks are composed of 5 fully-connected layers with 1,024 hidden units each. The baseline NES architecture\citep{marris2022turbocharging} consisted of 4 ``payoff to payoff'' layers with 128 channels, a ``payoff to dual'' layer with 256 channels and 4 ``dual to dual'' layers with 256 channels. Each layer uses mean and max pooling functions. All NfgTransformer model variants have $H=8$ attention heads. Parameter counts of all model variants are reported in Table~\ref{tab:params_count}. 

\begin{table}[ht]
\caption{The number of network parameters by configuration for each task. We note that the number of parameters of the NfgTransformer and the NES is independent from the size of the games. This is in contrast to fully-connected networks whose parameter counts depend on the input sizes.\label{tab:params_count}}
\centering
\begin{small}
\begin{tabular}{l|rr}
Model                                 & \shortstack[r]{\# Parameter \\ (NE)} & \shortstack[r]{\# Parameter \\ (Max-Deviation-Gain)} \\ \hline
{\tt NfgTransformer(D=~32,K=2,A=1)}  & 0.15M             & 0.16M                    \\ %
{\tt NfgTransformer(D=~32,K=4,A=1)}  & 0.31M             & 0.31M                    \\ %
{\tt NfgTransformer(D=~32,K=8,A=1)}  & 0.61M             & 0.62M                    \\ %
{\tt NfgTransformer(D=~64,K=8,A=0)}  & 1.10M              & 1.11M                    \\ %
{\tt NfgTransformer(D=128,K=2,A=1)} & 1.22M             & 1.29M                    \\ %
{\tt NfgTransformer(D=~64,K=8,A=1)}  & 1.63M             & 1.64M                    \\ %
{\tt NfgTransformer(D=~64,K=8,A=2)}  & 2.16M             & 2.17M                    \\ %
{\tt NfgTransformer(D=128,K=4,A=1)} & 2.44M             & 2.51M                    \\ %
{\tt NfgTransformer(D=128,K=8,A=1)} & 4.88M             & 4.95M                    \\ \hline
{\tt EquivariantMLP}                  & 4.76M - 16.83M    & 4.99M - 20.98M           \\ %
{\tt MLP}                             & 4.76M - 16.83M    & 4.99M - 20.98M           \\ %
{\tt NES}                             & 2.25M   & 2.51M          \\ %
\end{tabular}
\end{small}
\end{table}

\paragraph{Convergence progression}

Figure~\ref{fig:synthetic_games} visualises the training progression of each model configuration, task and game size from the same experiments reported in Table~\ref{tab:synthetic_games}.

\subsection{Payoff Prediction in {\sc Disc} Games}
\label{app:disc_ranking}

\paragraph{Game Sampling} Following Definition~\ref{def:disc}, generating {\sc Disc} games amounts to sampling latent vectors $\vu_t, \vv_t \in \sR^{Z}, t \in [T]$. Any real-valued latent vectors would define a valid {\sc Disc} game and we let $\vu_t = \vn + u$ with $\vn \sim \mathcal{N}(\mathbf{0}, \mathbf{1})$ and $u \sim \mathcal{U}(-1, 1)$. We sample $\vv_t$ in the same way. The shift random variable $u$ is not strictly necessary in this case, but it increases the probability that the resulting {\sc Disc} game is {\em not} fully cyclic following Proposition 1 of \citet{bertrand2023limitations}.

\paragraph{Masking} For each sampled instance of the DISC game, with a payoff tensor of shape $[N, T, \dots, T]$, we additionally sample a binary mask of shape $[T, \dots, T]$ where each element follows $\text{Bernoulli}(p)$. Both the game payoff tensor and the sampled mask for the game tensor are provided as inputs to the NfgTransformer network. We ensure that the model does not observe the payoff values of masked joint-actions following Equation~\ref{eq:attention_mask}. During loss computation, we minimise the $L_2$ loss (Figure~\ref{fig:nfg_transformer_decoder} (Right)) over all joint-actions, observed (i.e. for reconstruction) or unobserved (i.e. for prediction).

\paragraph{Architecture} For all results in this section, we used {\tt NfgTransformer(K=8,A=1,D=64)} with $H = 8$ attention heads for all attention operations.

\paragraph{Baseline Solvers} For all baseline results, we used the open-source implementation of Elo, mElo and xElo of released at \url{https://github.com/QB3/discrating}. For mElo and xElo, we used {\tt n\_components = 3} and the same settings as reported in \citet{bertrand2023limitations}.

\section{Interpretability Results}
\label{app:gamut}

We provide additional details on the empirical results in Section~\ref{sec:interpretation}. Figure~\ref{fig:gamut_ne} shows that despite simplifications made in Section~\ref{sec:interpretation} for our interpretability results, the NfgTransformer remains capable of equilibrium-solving in most games to reasonable accuracy, with {\tt CovariantGame (d06)} the most challenging game class. We show a failure case in this game class in Figure~\ref{fig:additional_interp} (Middle) and present additional example instances where the model successfully solved for a mixed-strategy NE (Top) or generalised to the out-of-distribution game class of Blotto \citep{roberson2006colonel}. Please refer to figure caption for additional remarks on the results.

\begin{figure}
    \centering
    \includegraphics[width=\textwidth]{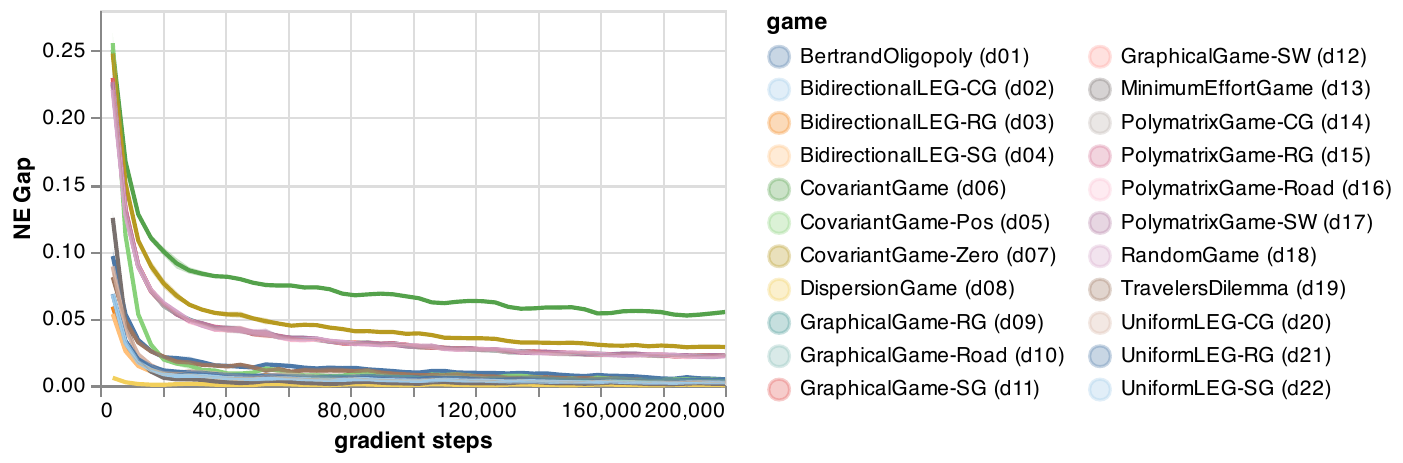}
    \caption{{\sc NE Gap} reported for each of the 22 GAMUT games throughout training. We note that a single network, {\tt NfgTransformer(K=8,A=0,D=16)}, with $H = 1$ is optimised to solve for Nash Equilibrium across all game classes.}
    \label{fig:gamut_ne}
\end{figure}

\begin{figure}
    \centering
    \includegraphics[width=\textwidth]{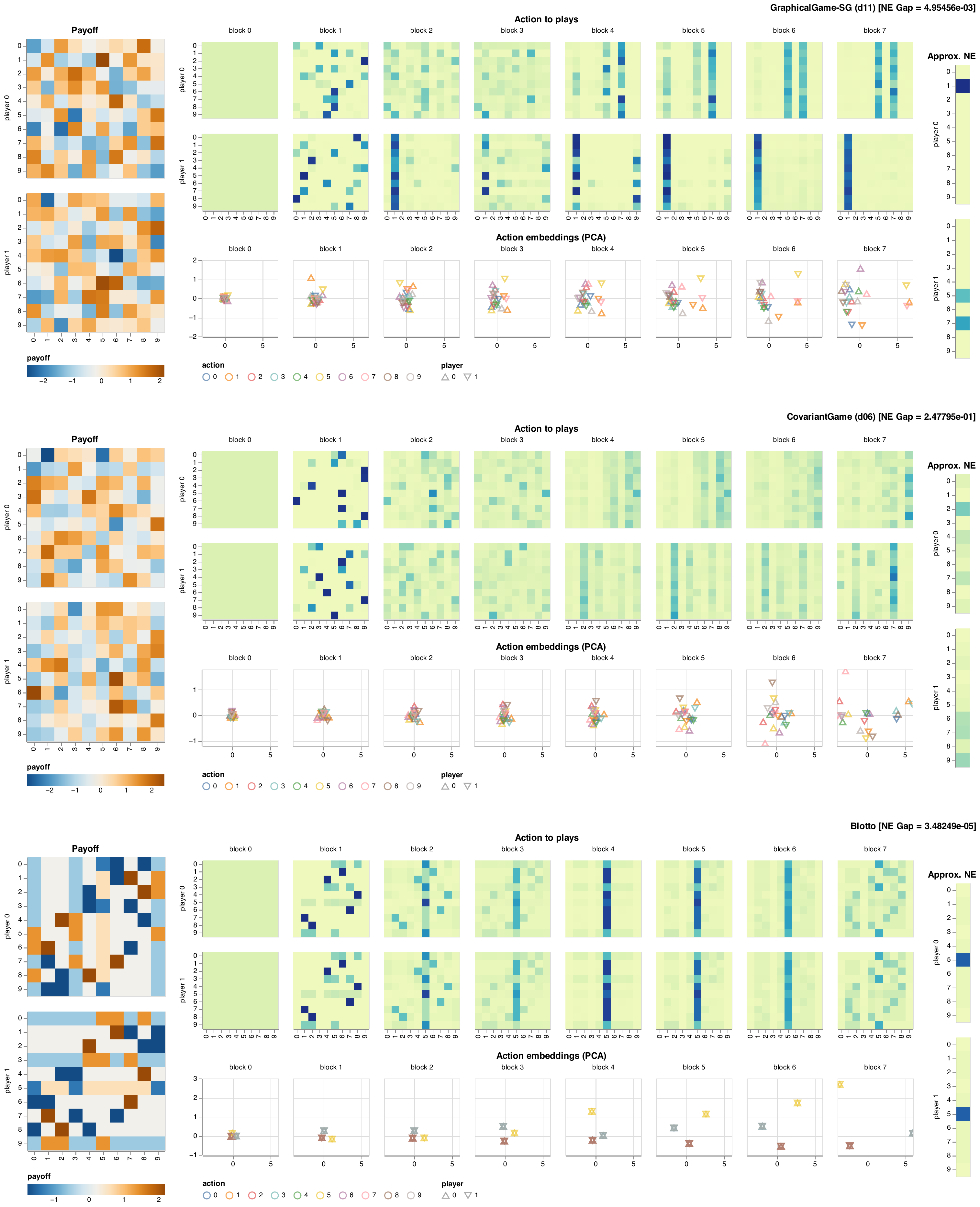}
    \caption{Here we provide additional interpretability results similar to Figure~\ref{fig:gamut} but for games that are asymmetric, have mixed-strategy NE (Top) or out-of-distribution (Bottom). We additionally provide an instance where the model struggled to find an NE (Middle) where the attention masks did not appear to have converged. For Blotto \citep{roberson2006colonel} which is a game class not seen during training, the model generalised well and identified a pure-strategy NE. The action embeddings also revealed three clusters, corresponding to the three strategically equivalent classes of actions. Note that one of the clusters corresponds to the dominant action of the two players.}
    \label{fig:additional_interp}
\end{figure}

\section{The Space of 2×2 Games}
\label{app:two_by_two}

\cite{marris2023equilibrium} introduced a subset of 2×2 normal-form games that any 2×2 game can be mapped to without changing its set of (coarse) correlated equilibria and Nash equilibria. This subset of games is called the equilibrium-invariant subset, and includes all possible nontrivial strategic interactions of 2×2 games.
Properties of games such as their equilibria, permutation symmetries, and best-response dynamics can be visualized in this ``map of games''. We can analyse the embeddings found by the NfgTransformer by sweeping over the nontrivial 2×2 equilibrium-invariant subset.

We used the transformer architecture {\tt NfgTransformer(K=2,A=1,D=16)} with $H=2$ attention heads at every self- and cross-attention layer. We used an additional linear layer to reduce the action embedding dimension down to 1, per player, per action, resulting in four variables to describe the game embeddings.
We trained NfgTransformer with an NE objective over the space of equilibrium-invariant subsets, and verified that the loss approaches zero. With the trained NfgTransformer, we sweep over the nontrivial 2×2 equilibrium-invariant subset, and visualize the embeddings (Figure~\ref{fig:embeddings_in_2x2_space}).

\begin{figure}
    \centering
    \begin{subfigure}[t]{0.49\linewidth}
        \embedding[image=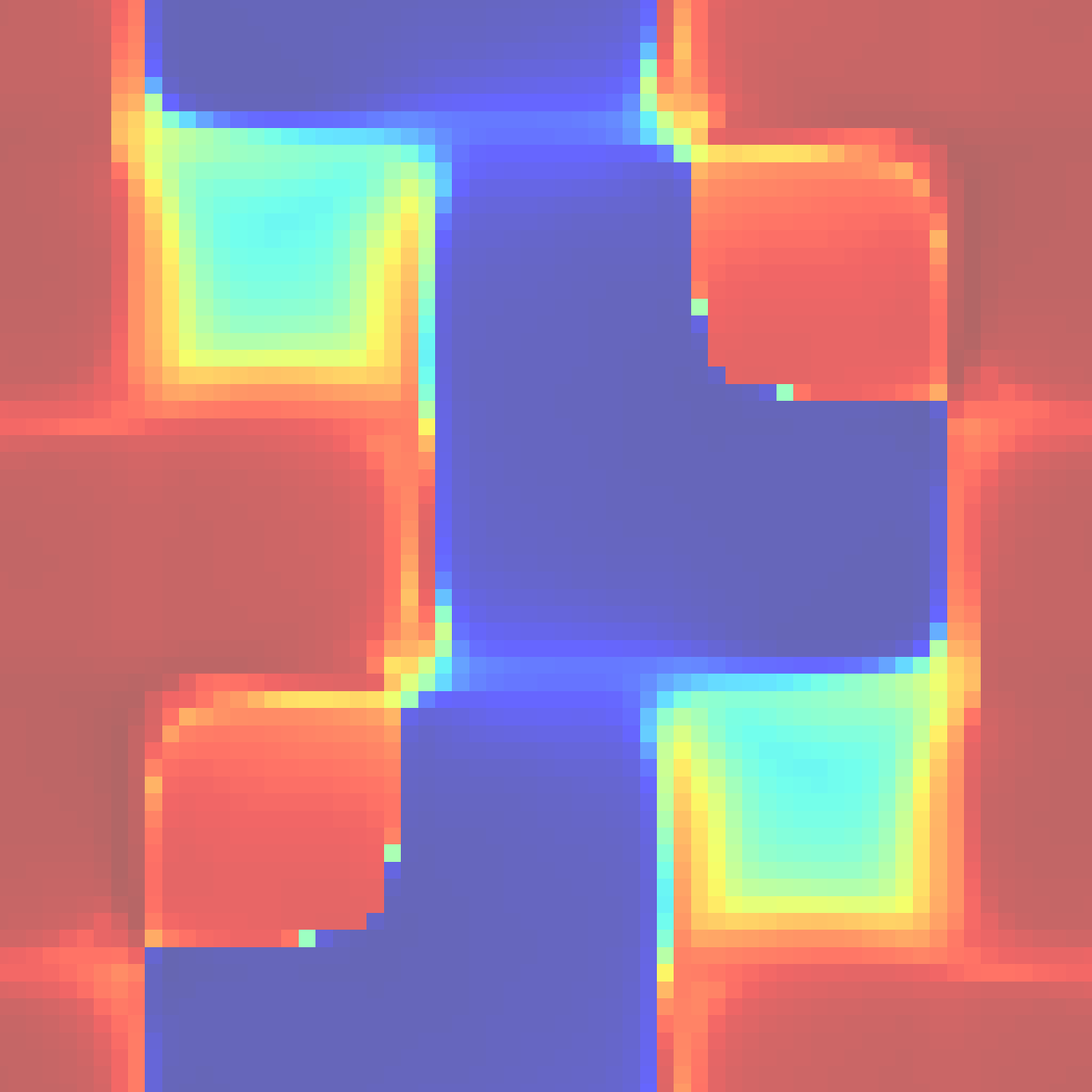,no axes labels,no tick labels,no best-response names]
        \caption{$\va^1_1$}
    \end{subfigure}\hfill
    \begin{subfigure}[t]{0.49\linewidth}
        \embedding[image=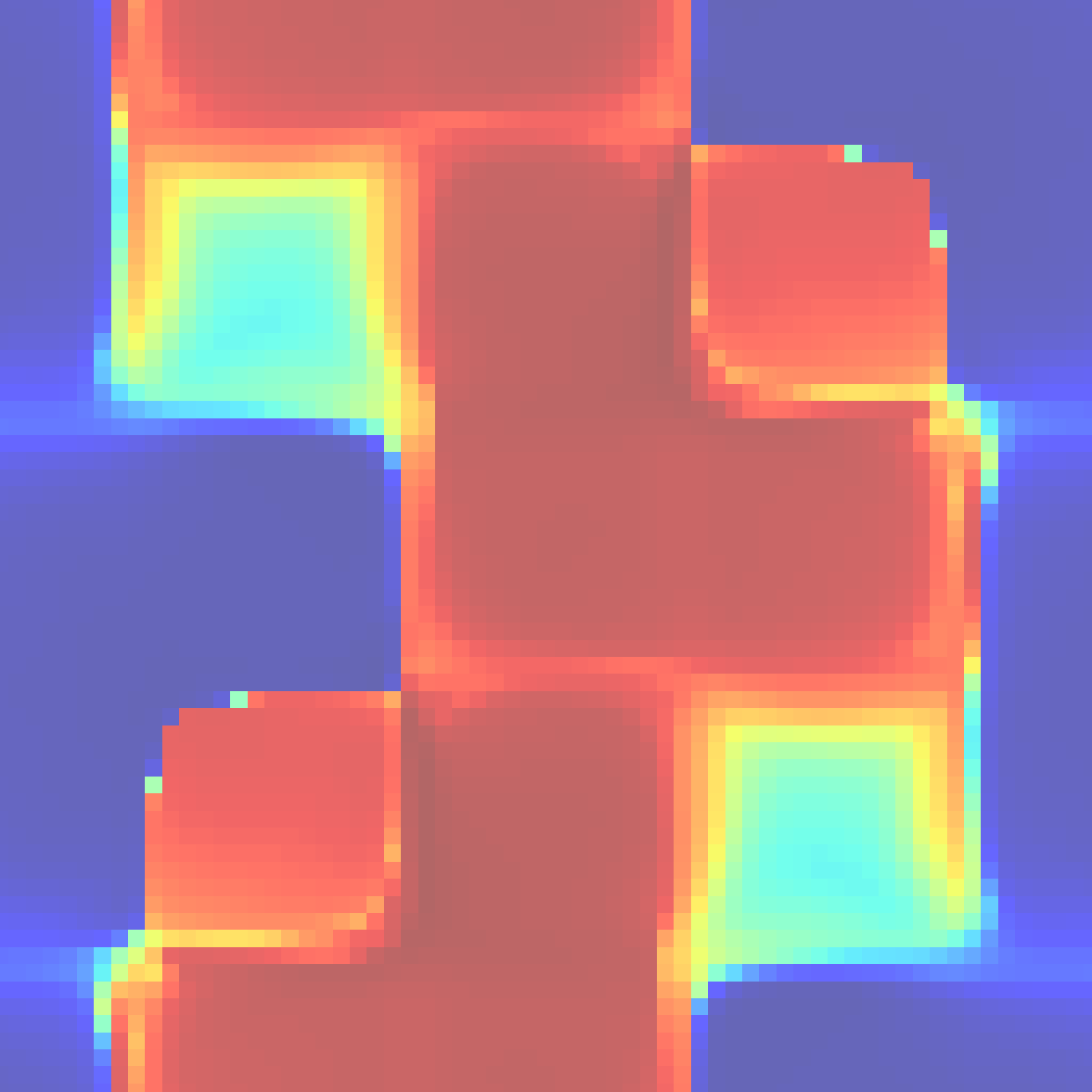,no axes labels,no tick labels,no best-response names]
        \caption{$\va^2_1$}
    \end{subfigure}

    \begin{subfigure}[t]{0.49\linewidth}
        \embedding[image=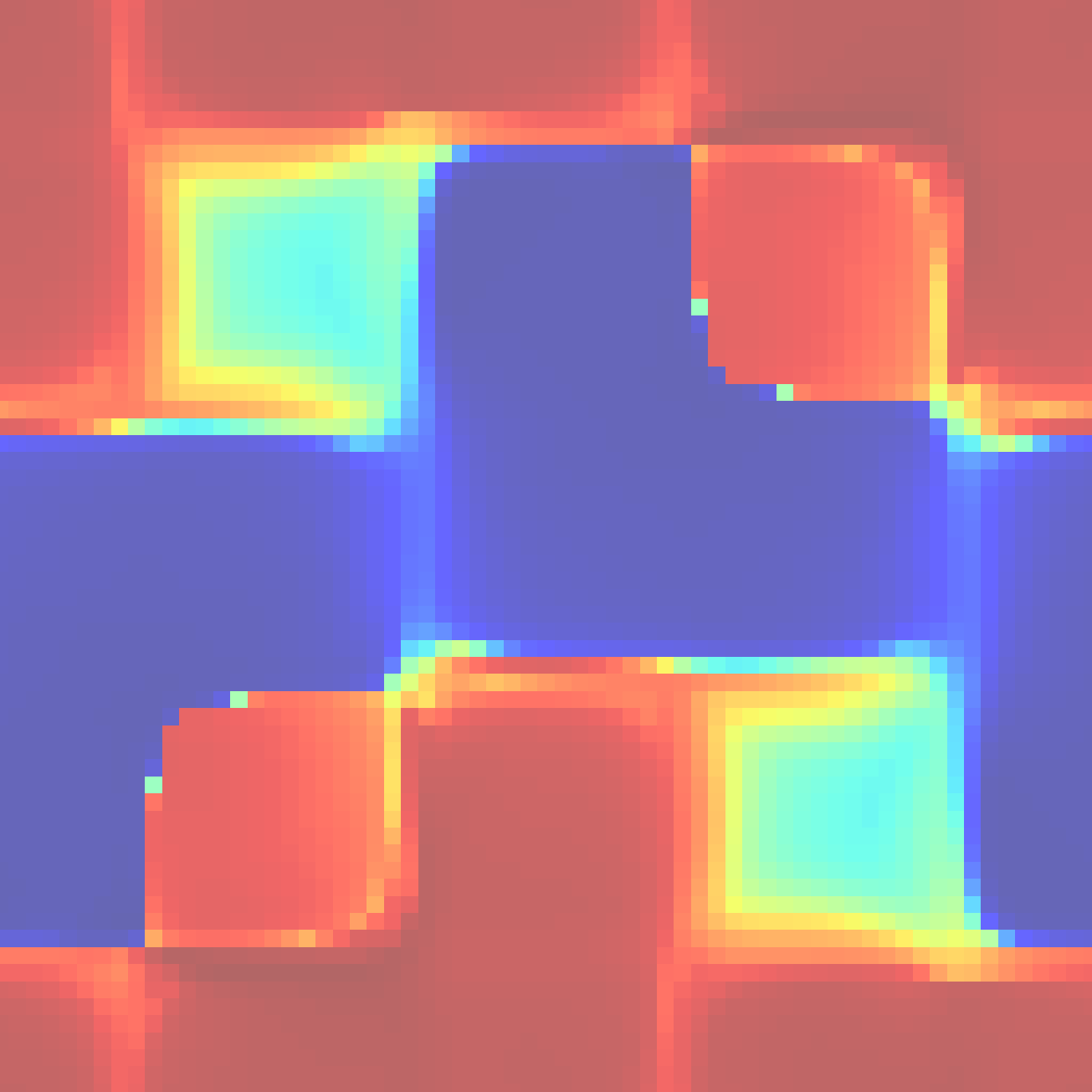,no axes labels,no tick labels,no best-response names]
        \caption{$\va^1_2$}
    \end{subfigure}\hfill
    \begin{subfigure}[t]{0.49\linewidth}
        \embedding[image=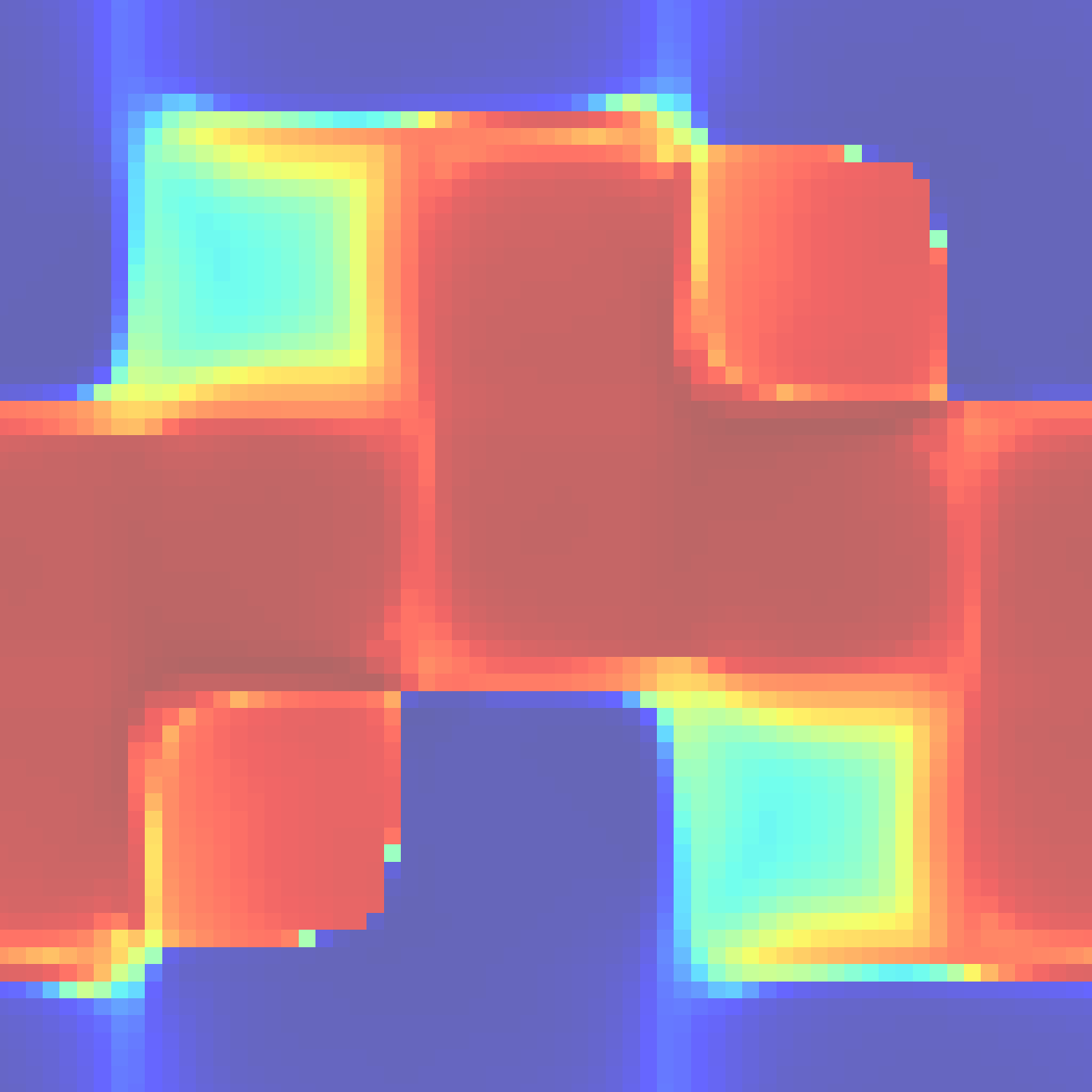,no axes labels,no tick labels,no best-response names]
        \caption{$\va^2_2$}
    \end{subfigure}

    \caption{NfgTransformer action-embeddings over the set of nontrivial 2×2 equilibrium-invariant normal-form games, when trained with an NE objective. The embeddings found closely follow the equilibrium boundaries (dark lines). Symmetries over the space of games are respected. Symmetric games (bottom-left to top-right diagonal) have the same embeddings between players. Permutations over players (folding over the bottom-left to top-right diagonal) are consistent. Colorbar: [$-11.58$ \raisebox{-0.18em}{\includegraphics[height=1em]{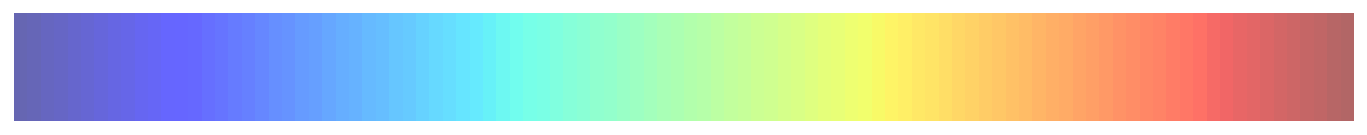}} $+11.56$].}
    \label{fig:embeddings_in_2x2_space}
\end{figure}

The learned action-embeddings have a very low value (\textcolor{blue}{blue} regions) when that action has all the mass in the NE, and very high value when the action has no mass in the NE (\textcolor{darkred}{dark red} regions). These ``L'' regions therefore correspond to games which have a single pure NE. The embeddings are low value (\textcolor{cyan}{cyan} regions) when the game has a mixed NE solution and occurs near the cyclic games (\brgraph{1}{0}{0}{1}{0}{1}{1}{0} and \brgraph{0}{1}{1}{0}{1}{0}{0}{1}). In these regions, all embeddings have to be similarly colored, as all actions are mixed. The embeddings are high value (\textcolor{red}{red} regions) in coordination game areas where there are two disconnected NEs (\brgraph{1}{0}{0}{1}{1}{0}{0}{1} and \brgraph{0}{1}{1}{0}{0}{1}{1}{0}). The borders between these regions correspond to changes in game payoffs when one action becomes become profitable than another, and as a result the NE can change drastically, and therefore so does the embedding.

When a game is symmetric, $G_1(a_1,a_2) = G_2(a_2,a_1)$, the embeddings between players are equal. We can verify this by studying the bottom-left to top-right diagonal. When swapping the player orders, we expect the embeddings to be swapped. Swapping players is equivalent to folding over the same diagonal. Again, we can visually verify that the embeddings are swapped. 

Next, we turn to the question of when the embeddings uniquely describe a game. We define a distance metric between action embeddings for game $i$ and game $j$, $d(\bA^i,\bA^j) = (\sum_{p \in [1, 2]} \sum_{a_p \in [1, 2]} ( \bA^i_p(a_p) - \bA^j_p(a_p) )^2)^\frac{1}{2}$, where $i,j \in \mathcal{G}$ are games sampled from a grid, which describes how close the embeddings of two games are to each other. We can also define the distance to the nearest other game within the set of considered games, $d_\text{min}(\bA^i, \mathcal{G}) = \min_{j \neq i \in \mathcal{G}} d(\bA^i,\bA^j)$. Using these distance metrics we can verify that $d_\text{min}(\bA^i,\bA^j) > 0$ apart from games with strategically equivalent actions (Figure~\ref{fig:embedding_dists}).

\begin{figure}
    \noindent\begin{minipage}{0.49\linewidth}
        \begin{subfigure}[t]{\linewidth}
            \embedding[image=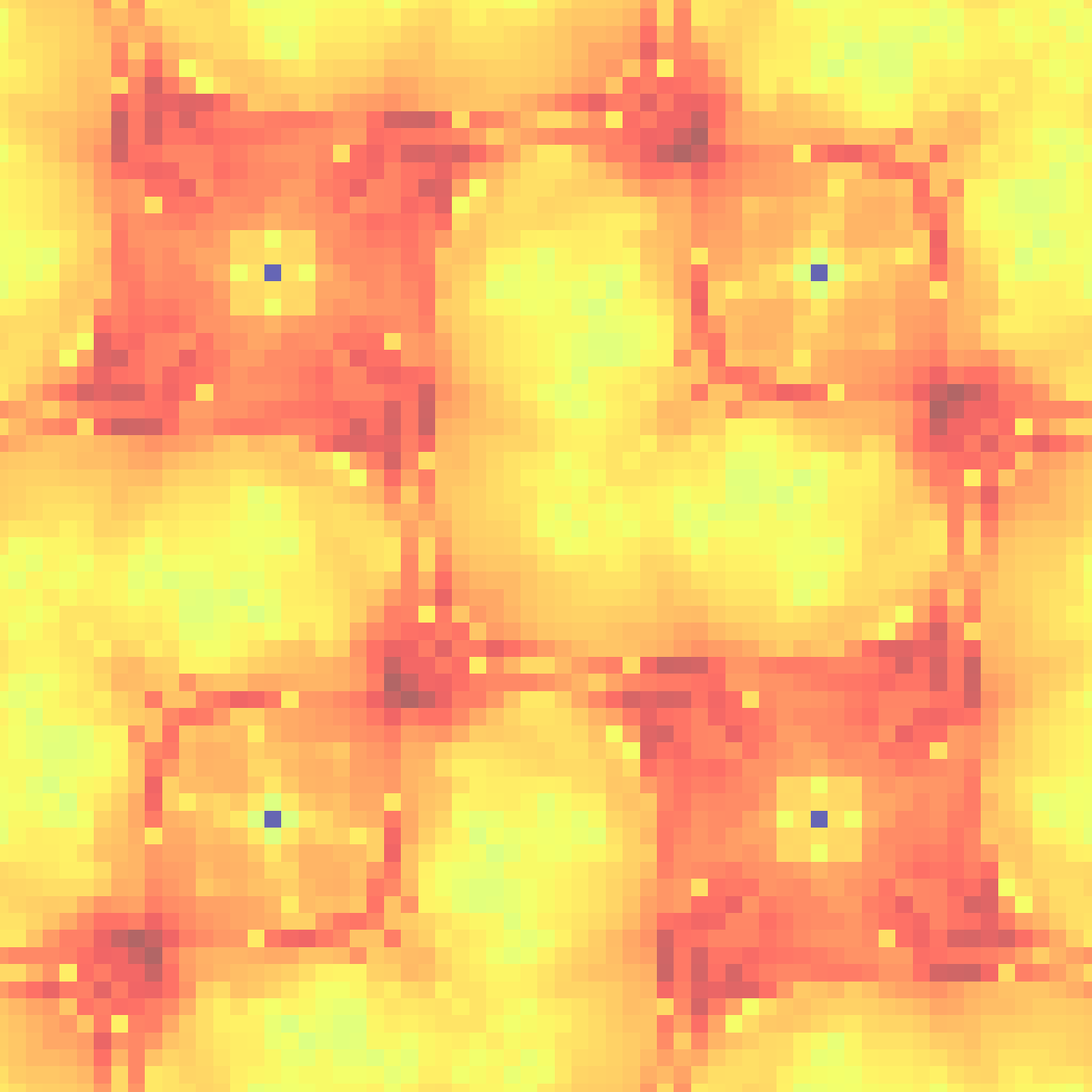,no axes labels,no tick labels,no best-response names]
            \caption{$\ln(d_\text{min}(\bA^i, \mathcal{G}))$}
            \label{fig:embedding_dist}
        \end{subfigure}
    \end{minipage}\hfill
    \begin{minipage}{0.49\linewidth}
        \begin{subfigure}[t]{0.49\linewidth}
            \centering
            \brpayoffstable[%
                row player first strategy label=$a_1^1$,%
                row player second strategy label=$a_1^2$,%
                column player first strategy label=$a_2^1$,%
                column player second strategy label=$a_2^2$,%
            ]{1}{0}{0}{1}{1}{0}{0}{1}
            \caption{Coord}
        \end{subfigure}
        \begin{subfigure}[t]{0.49\linewidth}
            \centering
            \brpayoffstable[%
                row player first strategy label=$a_1^1$,%
                row player second strategy label=$a_1^2$,%
                column player first strategy label=$a_2^1$,%
                column player second strategy label=$a_2^2$,%
            ]{1}{0}{0}{1}{0}{1}{1}{0}
            \caption{Cycle}
        \end{subfigure}
        
        \begin{subfigure}[t]{0.49\linewidth}
            \centering
            \brpayoffstable[%
                row player first strategy label=$a_1^1$,%
                row player second strategy label=$a_1^2$,%
                column player first strategy label=$a_2^1$,%
                column player second strategy label=$a_2^2$,%
            ]{0}{1}{1}{0}{0}{1}{1}{0}
            \caption{Anti Coord}
        \end{subfigure}
        \begin{subfigure}[t]{0.49\linewidth}
            \centering
            \brpayoffstable[%
                row player first strategy label=$a_1^1$,%
                row player second strategy label=$a_1^2$,%
                column player first strategy label=$a_2^1$,%
                column player second strategy label=$a_2^2$,%
            ]{0}{1}{1}{0}{1}{0}{0}{1}
            \caption{Anti Cycle}
        \end{subfigure}
        
        \vspace{1cm}
        
        \begin{subfigure}[t]{0.49\linewidth}
            \centering
            \begin{tabular}{c|cc}
                $\bA$ & $a_1$ & $a_2$ \\\hline
                $a^1$ & 9.415 & 9.415 \\
                $a^2$ & 9.415 & 9.415 \\
            \end{tabular}
            \caption{Coord Embedding}
        \end{subfigure}
        \begin{subfigure}[t]{0.49\linewidth}
            \centering
            \begin{tabular}{c|cc}
                $\bA$ & $a_1$ & $a_2$ \\\hline
                $a^1$ & -2.497 & -2.497 \\
                $a^2$ & -2.497 & -2.497 \\
            \end{tabular}
            \caption{Cycle Embedding}
        \end{subfigure}
    \end{minipage}

    \centering
    \caption{Subfigure \ref{fig:embedding_dist} shows the distance to the nearest other game embedding. The embeddings produced by NfgTransformer uniquely describe the 2×2 game apart from two edge cases. Two Coordination games (\brgraph{1}{0}{0}{1}{1}{0}{0}{1} and \brgraph{0}{1}{1}{0}{0}{1}{1}{0}) have identical embeddings, and two Cycle games (\brgraph{1}{0}{0}{1}{0}{1}{1}{0} and \brgraph{0}{1}{1}{0}{1}{0}{0}{1}) have identical embeddings, each because there are strategically equivalent. Colorbar: [$0.0$ \raisebox{-0.18em}{\includegraphics[height=1em]{figures/colormap.png}} $4.675$].}
    \label{fig:embedding_dists}
\end{figure}

The \brgraph{1}{0}{0}{1}{1}{0}{0}{1}~Coordination game has identical action embeddings to the \brgraph{0}{1}{1}{0}{0}{1}{1}{0}~Anti-Coordination game. In this case, due to permutation equivariance, the embedding for each action, in natural language, is: ``there is an action that the opponent can play which will give us both identical high payoff, and there is an action that the opponent can play which will give us both identical low payoff''. Due to the equivariant property it is not possible to disambiguate between these games from the embeddings alone. By initializing the network with action labels, hinting a reconstruction method with a row of true payoffs, or permuting the payoffs by a tiny amount, would all enable disambiguation. The last strategy can be seen from the figure where slightly biased coordination games all have positive distance to their nearest other game embedding. Similarly, the \brgraph{1}{0}{0}{1}{0}{1}{1}{0}~Cycle game (also known as matching pennies) has identical embeddings to the \brgraph{0}{1}{1}{0}{1}{0}{0}{1}~Anticlockwise Cycle game. In this case, the embedding is ``there is an action that the opponent can play which will me a high positive payoff and the opponent a high negative payoff, and there is an action that the opponent can play which will give me a high negative payoff and the opponent a high positive payoff''. Note that the Coordination and Cycle game have distinct embeddings. These are the only 4 points in the space that can only by disambiguated up to handedness. These appear with measure zero in the equilibrium-invariant subspace.

Overall, the embeddings neatly describe and predict the known structure of 2×2 games. The theoretically predicted properties, including permutation symmetries and NE, are reproduced in this experiment.

\end{document}

%% file: math_commands.tex
\usepackage{amsmath,amsfonts,bm}

\def\eqref#1{equation~\ref{#1}}

\def\1{\bm{1}}

\def\va{{\bm{a}}}

\def\vm{{\bm{m}}}
\def\vn{{\bm{n}}}

\def\vp{{\bm{p}}}

\def\vu{{\bm{u}}}
\def\vv{{\bm{v}}}

\DeclareMathAlphabet{\mathsfit}{\encodingdefault}{\sfdefault}{m}{sl}
\SetMathAlphabet{\mathsfit}{bold}{\encodingdefault}{\sfdefault}{bx}{n}

\def\sR{{\mathbb{R}}}

\newcommand{\expt}{\mathop{\mathbb{E}}}

\newcommand{\bA}{\mathbf{A}}
\newcommand{\bR}{\mathbb{R}}

\newcommand{\cA}{\mathcal{A}}

\newcommand{\notp}{{\neg p}}

\usepackage{amsmath}
\usepackage{amssymb}
\usepackage{mathtools}
\usepackage{amsthm}
\usepackage{xspace}

\theoremstyle{plain}
\newtheorem{theorem}{Theorem}[section]
\newtheorem{proposition}[theorem]{Proposition}

\theoremstyle{definition}
\newtheorem{definition}[theorem]{Definition}

\theoremstyle{remark}